\documentclass[conference,10pt]{IEEEtran}
\usepackage[usenames,dvipsnames]{xcolor}

\usepackage{xifthen}
\usepackage[utf8]{inputenc}
\usepackage{dashbox}
\usepackage{ebproof}
\usepackage{amsmath}
\usepackage{ebproof}
\usepackage{virginialake}
\usepackage{cmll}

\usepackage{amsthm}\theoremstyle{plain}
\newtheorem{thm}{Theorem}
\newtheorem{lemma}[thm]{Lemma}

\theoremstyle{definition}
\newtheorem{definition}[thm]{Definition}

\newtheorem{example}[thm]{Example}
\newtheorem{remark}[thm]{Remark}

\makeatletter
 \newcommand{\coloneqq}{%
 \mathrel{%
   \rlap{\raisebox{0.3ex}{$\m@th\cdot$}}\raisebox{-0.3ex}{$\m@th\cdot$}%
   \rlap{\raisebox{0.3ex}{$\m@th\cdot$}}\raisebox{-0.3ex}{$\m@th\cdot$}%
 }=}
\makeatother

\newcommand{\dual}[1]{\overline{#1}}
\newcommand{\cneg}[1]{\dual{#1}}

\newcommand{\VAR}{\textsc{var}}
\newcommand{\FUN}{\textsc{fun}}
\newcommand{\PRED}{\textsc{pred}}
\newcommand{\ATOM}{\textsc{atom}}
\newcommand{\FORM}{\textsc{form}}
\newcommand{\TERM}{\textsc{term}}

\newcommand{\samecol}{\sim}
\newcommand{\fequ}{\equiv}

\newcommand{\graph}[1]{\mathcal{#1}}
\newcommand{\vertices}[1][]{\ifthenelse{\isempty{#1}}{V}{V_{{\graph{#1}}}}}
\newcommand{\edges}[1][]{\ifthenelse{\isempty{#1}}{E}{E_{\graph{#1}}}}
\newcommand{\bgraph}[1]{\mathcal{\vec{#1}}}
\newcommand{\lgraph}[1]{\mathcal{#1}^{\mathsf{L}}}
\newcommand{\leaps}[1]{L_{#1}}

\newcommand{\vgraphof}[1]{V_{\graphof{#1}}}
\newcommand{\gA}{\graph{A}}
\newcommand{\gB}{\graph{B}}

\newcommand{\gC}{\graph{C}}
\newcommand{\gD}{\graph{D}}
\newcommand{\gG}{\graph{G}}
\newcommand{\gH}{\graph{H}}

\newcommand{\bD}{\bgraph{D}}
\newcommand{\bG}{\bgraph{G}}
\newcommand{\bH}{\bgraph{H}}

\newcommand{\vB}{\vertices[B]}

\newcommand{\vC}{\vertices[C]}
\newcommand{\vG}{\vertices[G]}
\newcommand{\vH}{\vertices[H]}

\newcommand{\eG}{\edges[G]}
\newcommand{\eH}{\edges[H]}

\newcommand{\sysS}{\mathsf{S}}
\newcommand{\Deri}{\Phi}
\newcommand{\DDeri}{\Psi}
\newcommand{\Proof}{\Pi}

\newcommand{\lgC}{\lgraph{\gC}}

\newcommand{\lpC}{\leaps{\gC}}

\ifCLASSINFOpdf
\else
\fi

\newcommand*{\FOLK}{\mathsf{LK1}}
\newcommand*{\FOLKcut}{\FOLK\mathord+\cut}

\newcommand*{\FOMLL}{\mathsf{MLL1^X}} 
\newcommand*{\FOMLLcut}{\FOMLL\mathord+\cut}

\newcommand*{\FOKS}{\mathsf{KS1}}
\newcommand*{\FOMLS}{\mathsf{MLS1^X}}

\newcommand{\rr}{\mathsf{r}}
\newcommand{\ax}{\mathsf{ax}}
\newcommand{\cut}{\mathsf{cut}}

\newcommand{\mix}{\mathsf{mix}}

\newcommand{\axr}{\mathsf{ax}}
\newcommand{\cutr}{\mathsf{cut}}
\newcommand{\mixr}{\mathsf{mix}}
\newcommand{\conr}{\mathsf{ctr}}
\newcommand{\weakr}{\mathsf{wk}}

\newcommand\aiD {\mathsf{ai}}
\newcommand\faD {\forall}
\newcommand\exD {\exists}
\newcommand\tttD {\ttt}
\renewcommand\wD {\mathsf{w}}

\newcommand\wrD {\mathsf{w}}
\renewcommand\cD {\mathsf{c}}
\renewcommand\acD {\mathsf{ac}}
\newcommand\acDx {\mathsf{ac}_x}
\newcommand\acDeq {\mathsf{ac}_x^\fequ}
\newcommand\wfaD {\mathsf{w_\forall}}
\newcommand\cfaD {\mathsf{c_\forall}}
\newcommand\mfaD {\mathsf{m_\forall}}
\newcommand\mexD {\mathsf{m_\exists}}

\newcommand{\cons}[1]{\{#1\}}
\newcommand{\Scons}[1]{S\cons{#1}}
\newcommand{\conhole}{\cons{\cdot}}
\newcommand{\Sconhole}{S\conhole}

\newcommand{\cor}{\vee}
\newcommand{\cand}{\wedge}

\newcommand{\PE}[1]{#1^\circ}

\newcommand\fv{\textsf{\small fv}}

\DeclareTextFontCommand{\bfit}{\bfseries\itshape}

\newcommand{\Pfour}{\mathbf{P_4}}

\newcommand{\tuple}[1]{\langle#1\rangle}
\newcommand{\pair}[1]{(#1)}
\newcommand{\set}[1]{\{#1\}}
\newcommand{\sqn}[1]{\vdash#1}
\newcommand{\sqns}[1]{\vdash#1\phantom{\vdash}}

\newcommand{\single}[1]{\bullet#1}

\newcommand{\rectif}[1]{\widehat{#1}}

\newcommand{\fographof}[1]{\llbracket#1\rrbracket}
\newcommand{\graphof}[1]{\llbracket#1\rrbracket}

\newcommand{\frameof}[1]{#1^\star}

\newcommand{\sublist}[1]{[#1]}
\newcommand{\subst}[2]{#1/#2}
\newcommand{\ssubst}[2]{\sublist{\subst{#1}{#2}}}
\newcommand{\substof}[1]{\sigma_{\!#1}}
\newcommand{\rsubstof}[1]{\rho_{#1}}
\newcommand{\dsubstof}[1]{\delta_{#1}}
\newcommand{\linkingof}[1]{\sim_{#1}}
\newcommand{\linking}{\sim}
\newcommand{\mapof}[1]{\lfloor{#1}\rfloor}
\newcommand{\labelof}[1]{\ell(#1)}
\newcommand{\form}[1]{\bigvee\mkern-1mu(\mkern-2mu #1\mkern-2mu)}

\renewcommand{\phi}{\varphi}

\newcommand{\quand}{\quad\mbox{and}\quad}
\newcommand{\qquand}{\qquad\mbox{and}\qquad}

\newcommand{\qquor}{\qquad\mbox{or}\qquad}

\usepackage{tikz}
\usetikzlibrary{arrows}
\usetikzlibrary{decorations}\usetikzlibrary{shapes.geometric}
\usetikzlibrary{chains}
\usetikzlibrary{arrows}
\usetikzlibrary{calc}
\usetikzlibrary{matrix}
\usetikzlibrary{positioning}
\usepackage{calculator}

\usepackage{pifont}

\usepackage{prftree}
\makeatletter
\newcommand\vctr[1]{\vcenter{\hbox{$\m@th{#1}$}}}
\makeatother

\newcommand\skewlifting[1]{\widetilde{#1}}

\newcommand\vertexa{w}

\newcommand\fib{\phi}

\renewcommand{\implies}{\Rightarrow}
\newcommand\tightwedgeveeshrink{\mkern-2mu}
\newcommand\tightvee{\mathbin{\tightwedgeveeshrink\vee\tightwedgeveeshrink}}
\newcommand\tightwedge{\mathbin{\tightwedgeveeshrink\wedge\tightwedgeveeshrink}}
\newcommand\shortimplies{\mathop{\mkern-0mu\implies\mkern-1mu}}
\newcommand\tightimplies{\shortimplies}

\newcommand\tightplus{\mathbin+}
\newcommand\tighttimes{\mathbin\times}
\newcommand\graphunion{\tightplus}
\newcommand\graphjoin{\tighttimes}

\newcommand\p{p}\newcommand\pp{\dual p}
\newcommand\q{q}\newcommand\qq{\dual q}
\newcommand\px{\p x}
\newcommand\py{\p y}

\newcommand\ppx{\pp x}

\newcommand\pa{\p a}
\newcommand\pb{\p b}
\newcommand\ppa{\pp a}
\newcommand\ppb{\pp b}
\newcommand\pz{\p z}

\newcommand\axpx{\forall x\mkern2.2mu\px}
\newcommand\aypy{\forall y \mkern3mu \py}
\newcommand\qab{\q a b}
\newcommand\qba{\q b a}
\newcommand\qqab{\qq a b}
\newcommand\qqba{\qq b a}

\newcommand\qqy{\qq y}
\newcommand\qxy{\q x y}
\newcommand\fx{f\mkern-1.5mu x}
\newcommand\fy{f\mkern-1.5mu y}
\newcommand\fz{f\mkern-1.5mu z}
\newcommand\pfx{\p\mkern-1mu \fx}
\newcommand\qfz{\q \fz}

\newcommand\pfy{\p\mkern-2mu \fy}
\newcommand\ffy{f\mkern-4mu \fy}
\newcommand\pffy{\p\mkern-2mu \ffy}
\newcommand\ppffy{\pp\mkern-2mu \ffy}
\newcommand\allx{\forall x}
\newcommand\ally{\forall y}
\newcommand\existsx{\exists x}

\newcommand\likex[1]{\raisebox{0ex}[1ex][0ex]{$#1$}}
\newlength\fheight
\newlength\pbarheight
\newcommand\likef[1]{\raisebox{0ex}[\the\fheight][0ex]{$#1$}}
\newcommand\likepbar[1]{\raisebox{0ex}[\the\pbarheight][0ex]{$#1$}}

\colorlet{dblue}{black!15!blue}
\colorlet{lblue}{blue!8}
\colorlet{dred}{black!55!red}
\colorlet{lred}{red!40}
\colorlet{dgreen}{black!65!green}
\colorlet{lgreen}{green!30}

\newcommand\roundvxwidth{3.6pt}

\tikzstyle{all vertices} = [inner sep=0pt, outer sep=0pt, draw]
\tikzstyle{round vx}     = [all vertices, circle, minimum width=\roundvxwidth, semithick]
\newcommand\squarewidth{4.8pt}
\tikzstyle{square vx}    = [all vertices, regular polygon, regular polygon sides = 4, minimum width=\squarewidth, minimum height=\squarewidth, semithick]
\tikzstyle{diamond vx}   = [square vx, shape border rotate = 45]

\tikzstyle{black vx}     = [round vx, draw=black, fill=black]
\tikzstyle{blue vx}      = [round vx, draw=dblue, fill=lblue]
\tikzstyle{red vx}       = [square vx, draw=dred, fill=lred]
\tikzstyle{green vx}     = [diamond vx, draw=dgreen, fill=lgreen]

\newcommand\vertextoitslabelgap{1.4pt}

\newcommand\edgepairvxgap{1.8ex}
\tikzset{
  graphlabel/.style={
	text height=1.4ex,
	text depth =.5ex,
	inner sep=0pt,
	outer sep=0pt,
        anchor=base
  }
}

\newcommand\vx[2]{
\node[black vx] ({#2}) at ({#1}) {};}

\newcommand\bluevx[2]{
\node[blue vx] ({#2}) at ({#1}) {};}

\newcommand\redvx[2]{
\node[red vx] ({#2}) at ({#1}) {};}

\newcommand\greenvx[2]{
\node[green vx] ({#2}) at ({#1}) {};}

\newcommand\collvxsep[6]{
\node[#6, label={[outer sep=0pt,inner sep=0,label distance={#5}]{#4}:${#3}$}] ({#2}) at ({#1}) {};}

\newcommand\lvxsep[5]{
\collvxsep{#1}{#2}{#3}{#4}{#5}{black vx}}
\newcommand\redlvxsep[5]{
\collvxsep{#1}{#2}{#3}{#4}{#5}{red vx}}
\newcommand\bluelvxsep[5]{
\collvxsep{#1}{#2}{#3}{#4}{#5}{blue vx}}

\newcommand\lvxgap{3pt}
\newcommand\lvx[4]{
\lvxsep{#1}{#2}{#3}{#4}{\lvxgap}}
\newcommand\redlvx[4]{
\redlvxsep{#1}{#2}{#3}{#4}{\lvxgap}}

\newcommand\lvxd[3]{
\lvxsep{#1}{#2}{#3}{-90}{\lvxgap}}

\newcommand\dummystart{\node (start) [graphlabel, inner sep=0, outer sep=0] {};}
\newcommand\vxsep[5][black]{
  \node (#2) [#1 vx, #5=#4 of #3] {};}
\newcommand\vxrightsep[4][black]{
  \vxsep[#1]{#2}{#3}{#4}{right}}

\newcommand\vxright[3][black]{
  \vxrightsep[#1]{#2}{#3}{\vertextoitslabelgap}}

\newcommand\labelright[4][\vertextoitslabelgap]{
  \node (#2) [graphlabel, right=#1 of #4] {$#3$};}

\newcommand\labelsolo[2]{
\dummystart\labelright[-1pt]{#1}{#2}{start}}

\newcommand\fullynamedcolouredpair[8]{
\labelsolo {#1} {#2}
\vxright[#4] {#3} {#1}
\vxrightsep[#6] {#5} {#3} {\edgepairvxgap}
\labelright {#7} {#8} {#5}
\e {#3} {#5}
}
\newcommand\fullynamedpair[6]{
  \fullynamedcolouredpair{#1}{#2}{#3}{black}{#4}{black}{#5}{#6}}
\newcommand\namedpair[4]{
  \fullynamedpair{leftlabel}{#1}{#2}{#3}{rightlabel}{#4}}

\newcommand\edgepair[2]{
  \namedpair{#1}{leftvx}{rightvx}{#2}}
\newcommand\edgepairpic[2]{
\begin{tikzpicture}[inlinegraph]\edgepair{#1}{#2}\end{tikzpicture}}

\newcommand\defaultarrowstyle{stealth'}
\tikzset{>=\defaultarrowstyle}

\newcommand\linewd{.5pt}

\newcommand\fibrewidth{1pt}
\newcommand\leapwidth{.7pt}
\newcommand\dualitywidth{.5pt}

\makeatletter
\tikzset{%
  stretch dash/.code args={on #1 off #2}{%
    \tikz@addoption{%
      \pgfgetpath\currentpath%
      \pgfprocessround{\currentpath}{\currentpath}%
      \pgf@decorate@parsesoftpath{\currentpath}{\currentpath}%
      \pgfmathparse{max(round((\pgf@decorate@totalpathlength-#1)/(#1+#2)),0)}%
      \let\npattern=\pgfmathresult%
      \pgfmathparse{\pgf@decorate@totalpathlength/(\npattern*(#1+#2)+#1)}%
      \let\spattern=\pgfmathresult%
      \pgfsetdash{{\spattern*#1}{\spattern*#2}}{0pt}%
    }%
  }%
}
\makeatother

\newcommand\e[2]{\draw ({#1})--({#2});}

\newcommand\eseps[4]{\draw[{_[sep=#3]}-{_[sep=#4]}] (#1) -- (#2);}

\newcommand\fe[2]{\draw ([yshift=-2.7pt]#1.south)--([yshift=2.7pt]#2.north);} 

\newcommand\deseps[5][\defaultarrowstyle]{\draw[{_[sep=#4]}-{>[sep=#5]},>=#1] (#2) -- (#3);}

\newcommand\de[2]{\deseps{#1}{#2}{2.4pt}{2.3pt}}

\makeatletter

\def\edgelen{.8}
\SQUAREROOT{.75}{\equitriangleheightmultiplier}
\MULTIPLY{\edgelen}{\equitriangleheightmultiplier}{\height}
\DIVIDE{\edgelen}{2}{\halfedgelen}
\DIVIDE{\halfedgelen}{2}{\quarteredgelen}
\DIVIDE{\quarteredgelen}{2}{\eighthedgelen}
\DIVIDE{\height}{2}{\halfheight}
\MULTIPLY{\edgelen}{2}{\twoedgelen}
\MULTIPLY{\edgelen}{3}{\threeedgelen}

\tikzstyle{tree} = [graph, level distance=4ex, outer sep=0pt, inner sep=0, 
  level 1/.style={sibling distance=7ex},
  level 2/.style={sibling distance=4ex},
  level 3/.style={sibling distance=4ex}
]
\tikzstyle{leaf} = [black vx, outer sep=2pt]

\newcommand\gr[1]{\vctr{\begin{tikzpicture}[graph]#1\end{tikzpicture}}}
\newcommand\tr[1]{\vctr{\begin{tikzpicture}[tree,
  level 1/.style={sibling distance=5ex},
  level 2/.style={sibling distance=5ex},
  level 3/.style={sibling distance=5ex}
]#1\end{tikzpicture}}}

\newcommand\namedleaflab[2] {  node[leaf, label={[label distance=1pt]{-90}:${#1}$}] (#2) {} } 
\newcommand\leaflab[1] {\namedleaflab{#1}{leaf}}

\newcommand\joinroot{\node[circle] {$\graphjoin$}}
\newcommand\unionnode{node[circle] {$\graphunion$}}

\tikzstyle{graph}        = [line width=\linewd]
\tikzstyle{fibres}       = 
 [every path/.style={graph, line cap=round, line width=\fibrewidth, stretch dash=on .001pt off 4.2pt}]
\tikzstyle{leap}         = [every path/.style={graph, line width=\leapwidth, stretch dash=on 5pt off 3pt}]
\tikzstyle{duality}      = 
 [every path/.style=
  {graph, line width=\dualitywidth, color=dualitycolour, stretch dash=on 3pt off 1.5pt}]
\tikzstyle{binding}      = [every path/.style={graph, line width=\dualitywidth, color=bindingcolour}]
\tikzstyle{semifibres}   = [graph]
\tikzstyle{inlinegraph} = [graph,baseline]

\newcommand\fibheight{1.2}
\newcommand\peircefibheight{1.7}
\newcommand\fibheighttall{1.5}

\def\lo{.13}%
\def\hi{.33}%
\def\vi{.53}

\tikzset{inlinecp/.style={graph,auto,inner sep=0pt,outer sep=0pt,node distance = 0pt}}
\tikzset{
  token/.style={
	text height=1.9ex,
	text depth =.5ex,
	inner sep=0pt,
	outer sep=0pt
  }
}
\newcommand\fibration[1]{
\begin{math}\begin{scope}[inlinecp,start chain=going right]#1\end{scope}\end{math}}

\newcommand\fibbase[2]{
\node[on chain,token] (#1) {${}#2{}$};}
\renewcommand\over[3]{\node[#1,above = #2 of \currentnode, anchor=center] (#3) {};}

\newcommand\fibre[3]{
\fibbase{#1}{#2}\def\currentnode{#1}#3}
\newcommand\symb[1]{\fibbase{symb}{#1}}

\newcommand\drinkerformula{\exists x\mkern1mu(\mkern1mu px\mkern-1mu\implies\mkern-1mu \aypy)}
\newcommand\veedrinkerformula{\exists x.(\cneg{p}x \vlor (\forall y.py))}

\newcommand\variantveedrinkerformula{\exists x\mkern2mu\forall y\mkern2mu(\py\mkern-1mu\vee\mkern-1mu\ppx)}

\newcommand\drinkergraph{\gD}

\newcommand\drinkerx{-1.1}
\newcommand\drinkerxx{-.4}
\newcommand\drinkerxxx{.4}
\newcommand\drinkerxxxx{1.15}

\newcommand\drinkerdown{.235}
\newcommand\drinkerdowndown{.37}
\newcommand\drinkercovergap{.445}
\newcommand\drinkercoverfirstdown{.3}
\newcommand\drinkercoverseconddown{.21}

\newcommand\drinkerbasevertices{
\lvxd{\drinkerx,0}{x}{x}
\lvxd{\drinkerxx,-\drinkerdowndown}{ppx}{\ppx}
\lvxd{\drinkerxxx,-\drinkerdown}{y}{y}
\lvxd{\drinkerxxxx,0}{py}{\py}
}

\newcommand\drinkerbaseedges{
\begin{scope}[graph]
\e x y
\e x {py}
\e x {ppx}
\end{scope}
}

\newcommand\drinkerbase{
\drinkerbasevertices
\drinkerbaseedges
}

\newcommand\drinkercoververticescoloured{
\begin{scope}[shift={(0,\drinkercovergap)}]
    \vx{\drinkerx,0}{xa}
    \vx{\drinkerxxx,-\drinkercoverseconddown}{ya}
    \bluevx{\drinkerxxxx,0}{pya}
  \end{scope}
  \vx{\drinkerx,0}{xb}
  \bluevx{\drinkerxx,-\drinkercoverfirstdown}{ppxb}
}

\newcommand\drinkercoveredges{
\begin{scope}[graph]
  \e{xb}{ppxb}
  \e{xa}{pya}
  \e{xa}{ya}
\end{scope}}

\newcommand\drinkercovercoloured{
  \drinkercoververticescoloured
  \drinkercoveredges
}

\newcommand\drinkerfibres{
\begin{scope}[fibres]
  \fe{xa}{xb}
  \fe{xb}{x}
  \fe{ppxb}{ppx}
  \fe{pya}{py}
  \fe{ya}{y}
\end{scope}
}

\newcommand\drinkerfibcoloured{
\begin{scope}[shift={(0,\fibheighttall)}]
  \drinkercovercoloured
\end{scope}
\drinkerbase
\drinkerfibres
}

\newcommand\drinkerfibcolouredpic{
  \begin{tikzpicture}
    \drinkerfibcoloured
  \end{tikzpicture}
}

\newcommand\pfyxone{-1.2}
\newcommand\pfyxtwo{-.6}
\newcommand\pfyxthree{0}
\newcommand\pfyxfour{.6}
\newcommand\pfyxfive{1.2}
\newcommand\pfycoverradius{.23}
\newcommand\pfycp{
  \begin{scope}[shift={(0,\fibheight)}]
    \begin{scope}[shift={(0,-\pfycoverradius)}]
      \vx{\pfyxone,0}{xa}
      \bluevx{\pfyxtwo,0}{ppxa}
    \end{scope}
    \begin{scope}[shift={(0,\pfycoverradius)}]
      \vx{\pfyxone,0}{xb}
      \redvx{\pfyxtwo,0}{ppxb}
    \end{scope}
    \vx{\pfyxthree,0}{yc}
    \bluevx{\pfyxfour,0}{pyc}
    \redvx{\pfyxfive,0}{pfyc}
  \end{scope}
  \lvxd{\pfyxone,0} x {\likepbar x}
  \lvxd{\pfyxtwo,0}{ppx}{\likepbar\ppx}
  \lvxd{\pfyxthree,0}{y}{\likepbar y}
  \lvxd{\pfyxfour,0}{py}{\likepbar\py}
  \lvxd{\pfyxfive,0}{pfy}{\mkern14mu\likepbar\pfy}
  \e {xa} {ppxa}
  \e {xb} {ppxb}
  \e {pyc} {pfyc}
  \e {x} {ppx}
  \e {py} {pfy}
  \begin{scope}[fibres]
    \fe{xb}{xa}
    \fe{xa}{x}
    \fe{ppxb}{ppxa}
    \fe{ppxa}{ppx}
    \fe{yc}{y}
    \fe{pfyc}{pfy}
    \fe{pyc}{py}
  \end{scope}
}

\newcommand\pabxone{-1.2}
\newcommand\pabxtwo{-.5}
\newcommand\pabxthree{.2}
\newcommand\pabxfour{.7}
\newcommand\pabxfive{1.2}
\newcommand\pabcoverradius{.3}
\newcommand\pabnudge{-.22}
\newcommand\pabcp{
  \begin{scope}[shift={(0,\fibheight)}]
    \bluevx{\pabxone,0}{pabm}
    \redvx{\pabxtwo,0}{pcdm}
    \begin{scope}[shift={(0,\pabcoverradius)}]
      \vx{\pabxthree,0}{xu}
      \vx{\pabxfour,\pabnudge}{yu}
      \bluevx{\pabxfive,0}{pxyu}
    \end{scope}
    \begin{scope}[shift={(0,-\pabcoverradius)}]     
      \vx{\pabxthree,0}{xl}
      \vx{\pabxfour,\pabnudge}{yl}
      \redvx{\pabxfive,0}{pxyl}
    \end{scope}
  \end{scope}
  \lvxd{\pabxone,0}{pab}{\hspace{-1ex}\likepbar\qqab}
  \lvxd{\pabxtwo,0}{pcd}{\likepbar\qqba\hspace{-.5ex}}
  \lvxd{\pabxthree,0}{x}{\likepbar x} 
  \lvxd{\pabxfour,\pabnudge}{y}{y}
  \lvxd{\pabxfive,0}{pxy}{\likepbar\qxy\hspace{-1ex}}
  \e {pabm} {pcdm}
  \e {xu} {yu}
  \e {xu} {pxyu}
  \e {yu} {pxyu}
  \e {xl} {yl}
  \e {xl} {pxyl}
  \e {yl} {pxyl}
  \e {pab} {pcd}
  \e {x} {y}
  \e {x} {pxy}
  \e {y} {pxy}
  \begin{scope}[fibres]
    \fe{pabm}{pab}
    \fe{pcdm}{pcd}
    \fe{xu}{xl}
    \fe{xl}{x}
    \fe{yu}{yl}
    \fe{yl}{y}
    \fe{pxyu}{pxyl}
    \fe{pxyl}{pxy}
  \end{scope}
}
\newcommand\pabcpinline{\fibration{
  \def\hi{0.43}
  \fibre{qabbase}{\qab}{
    \over{blue vx}{\lo}{pab}
  }
  \symb{\vee}
  \fibre{qbabase}{\qba}{
    \over{red vx}{\lo}{pcd}
  }
  \symb{\mkern5mu\implies\mkern4mu}
  \fibre{exbase}{\exists x}{
    \over{black vx}{\lo}{xl}
    \over{black vx}{\hi}{xu}
  }
  \symb{\mkern4mu}
  \fibre{eybase}{\exists y}{
    \over{black vx}{\lo}{yl}
    \over{black vx}{\hi}{yu}
  }
  \symb{\mkern3mu}
  \fibre{qxybase}{\qxy}{
    \over{red vx}{\lo}{pxyl}
    \over{blue vx}{\hi}{pxyu}
  }
  \e {pab} {pcd}
  \e {xl} {yl}
  \e {yl} {pxyl}
  \e {xu} {yu}
  \e {yu} {pxyu}
  \draw (xl) to[out=25,in=160,looseness=.9](pxyl);
  \draw (xu) to[out=25,in=160,looseness=.9](pxyu);
}}

\newcommand\onionxone{.3}
\newcommand\onionxtwo{.85}
\newcommand\onionxthree{1.4}
\newcommand\onionxfour{2.1}
\newcommand\onionxfive{2.8}
\newcommand\onionxsix{3.5}
\newcommand\onioncoverradius{.3}
\newcommand\onionnudge{-.22}
\newcommand\onioncp{
  \begin{scope}[shift={(0,\fibheight)}]
    \begin{scope}[shift={(0,-\onioncoverradius)}]
      \vx{\onionxone,0}{xa}
      \redvx{\onionxtwo,\onionnudge}{pfxa}
      \greenvx{\onionxthree,0}{ppxa}
    \end{scope}
    \begin{scope}[shift={(0,\onioncoverradius)}]
      \vx{\onionxone,0}{xb}
      \greenvx{\onionxtwo,\onionnudge}{pfxb}
      \bluevx{\onionxthree,0}{ppxb}
    \end{scope}
    \vx{\onionxfour,0}{yc}
    \redvx{\onionxfive,0}{ppffyc}
    \bluevx{\onionxsix,0}{pyc}
  \end{scope}
  \lvxd{\onionxone,0} x {\likepbar x\hspace{1ex}}
  \lvxd{\onionxtwo,\onionnudge}{pfx}{\pfx}
  \lvxd{\onionxthree,0}{ppx}{\hspace{1.5ex}\ppx}
  \lvxd{\onionxfour,0} y {\likepbar y}
  \lvxd{\onionxfive,0}{ppffy}{\likepbar\ppffy\hspace{.2ex}}
  \lvxd{\onionxsix,0}{py}{\hspace{.5ex}\likepbar\py}
  \e {xa} {pfxa}
  \e {xa} {ppxa}
  \e {pfxa} {ppxa}
  \e {xb} {pfxb}
  \e {xb} {ppxb}
  \e {pfxb} {ppxb}
  \e {x} {pfx}
  \e {x} {ppx}
  \e {pfx} {ppx}
  \begin{scope}[fibres]
    \fe{xb}{xa}
    \fe {xa}{x}
    \fe {pfxb}{pfxa}
    \fe {pfxa}{pfx}
    \fe {ppxb}{ppxa}
    \fe {ppxa}{ppx}
    \fe {yc}{y}
    \fe {ppffyc}{ppffy}
    \fe {pyc}{py}
  \end{scope}
}
\newcommand\onioncpinline{\fibration{
  \def\hi{.43}
  \symb{\left(\strut\mkern-5mu\right.}
  \fibre{axbase}{\allx}{
    \over{black vx}{\lo}{xl}
    \over{black vx}{\hi}{xu}
  }
  \symb{(}
  \fibre{pfxbase}{\pfx}{
    \over{red vx}{\lo}{pfxl}
    \over{green vx}{\hi}{pfxu}
  }
  \symb{\mkern-2mu\implies\mkern-2mu}
  \fibre{pxbase}{\px}{
    \over{green vx}{\lo}{pxl}
    \over{blue vx}{\hi}{pxu}
  }
  \symb{)\left.\mkern-5mu\strut\right)\mkern3mu\implies\mkern3mu}
  \fibre{aybase}{\ally}{
    \over{black vx}{\lo}{y}
  }
  \symb{\mkern2mu(\mkern1mu}
  \fibre{pffybase}{\pffy}{
    \over{red vx}{\lo}{pffy}
  }
  \symb{\mkern-1mu\implies\mkern-1mu}
  \fibre{pybase}{\py}{
    \over{blue vx}{\lo}{py}
  }
  \symb{\mkern1mu)}
  \e {xl} {pfxl}
  \e {pfxl} {pxl}
  \e {xu} {pfxu}
  \e {pfxu} {pxu}
  \draw (xu) to[out=23,in=170,looseness=.9](pxu);
  \draw (xl) to[out=23,in=170,looseness=.9](pxl);
}}

\newcommand\eabxone{-.8}
\newcommand\eabxtwo{-.3}
\newcommand\eabxthree{.3}
\newcommand\eabxfour{.825}
\newcommand\eabcoverradius{.22}
\newcommand\eaby{-.2}
\newcommand\eabyy{-.35}
\newcommand\eabfibheightboost{.12}
\newcommand\eabcp{
  \begin{scope}[shift={(0,\fibheight)}]
    \begin{scope}[shift={(0,\eabfibheightboost)}]
      \begin{scope}[shift={(0,-\eabcoverradius)}]
        \vx{\eabxone,0}{xl}
        \bluevx{\eabxtwo,\eabyy}{ppal}
        \redvx{\eabxthree,\eaby}{ppbl}
        \redvx{\eabxfour,0}{pxm}
      \end{scope}
      \begin{scope}[shift={(0,\eabcoverradius)}]
        \vx{\eabxone,0}{xu}
        \bluevx{\eabxfour,0}{pxu}
      \end{scope}
    \end{scope}
  \end{scope}
  \lvxd{\eabxone,0} x {x\hspace{1.3ex}}
  \lvxd{\eabxtwo,\eabyy}{ppa}{\likepbar\ppa\hspace{.2ex}}
  \lvxd{\eabxthree,\eaby}{ppb}{\hspace{.9ex}\likepbar\ppb}
  \lvxd{\eabxfour,0}{px}{\px\hspace{-1.5ex}}
  \e {xl} {ppal}
  \e {xl} {ppbl}
  \e {ppal} {ppbl}
  \e {xl} {pxm}
  \e {xu} {pxu}
  \e {x} {ppa}
  \e {x} {ppb}
  \e {ppa} {ppb}
  \e {x} {px}
  \begin{scope}[fibres]
    \fe {xu}{xl}
    \fe {xl}{x}
    \fe {ppal}{ppa}
    \fe {ppbl}{ppb}
    \fe {pxu}{pxm}
    \fe {pxm}{px}
  \end{scope}
}

\newcommand\pfyformula{(\axpx) \mkern1mu \implies \mkern1mu \forall y\mkern3mu (\mkern1mu\py\tightwedge \pfy\mkern1mu)}
\newcommand\eabformula{\exists x\mkern2mu (\mkern1mu\pa\tightvee\mkern1.2mu\pb\mkern-1mu\implies\mkern-1mu px\mkern1mu)}
\newcommand\onionformula{\left(\strut\forall x\mkern1mu  (\mkern1mu \pfx\tightimplies \px)\right)\mkern2mu \implies\mkern2mu  \forall y\mkern2mu  (\mkern1mu \pffy\tightimplies \py\mkern1mu )}
\newcommand\pabformula{\qab\vee \qba\mkern4mu \implies \mkern4mu \exists x\mkern4mu \exists y\mkern3mu \qxy}

\newcommand\cponformula[3]{
  \begin{scope}[graph,shift={(#2,1.3)}]{#1}\end{scope}
  \node[anchor=base] (formula) at (0,0){$#3$};
}

\newcommand\eabcponformula{\cponformula{\eabcp}{-.1}{\eabformula}}
\newcommand\pfycponformula{\cponformula{\pfycp}{-.15}{\pfyformula}}
\newcommand\pabcponformula{\cponformula{\pabcp}{0}{\pabformula}}
\newcommand\onioncponformula{\cponformula{\onioncp}{-2}{\onionformula}}

\newcommand\cponeonformula{\eabcponformula}
\newcommand\cptwoonformula{\pfycponformula}
\newcommand\cpthreeonformula{\pabcponformula}
\newcommand\cpfouronformula{\onioncponformula}

\newcommand\eabcpinline{\fibration{
  \def\hi{.39}
  \def\vi{.66}
  \fibre{exbase}{\existsx}{
    \over{black vx}{\lo}{x}
    \over{black vx}{\hi}{xm}
    \over{black vx}{\vi}{xu}
  }
  \symb{\mkern2mu(\mkern2mu}
  \fibre{pabase}{\pa}{
    \over{blue vx}{\lo}{pa}
  }
  \symb{\tightvee\mkern1.2mu}
  \fibre{pybase}{\pb}{
    \over{red vx}{\lo}{pb}
  }
  \symb{\mkern-1mu\implies\mkern-1mu}
  \fibre{pxbase}{\px}{
    \over{red vx}{\lo}{px}
    \over{blue vx}{\hi}{pxm}
  }
  \symb{\mkern2mu)}
  \e {x}{pa}
  \e {pa}{pb}
  \draw (x) to[out=20,in=165,looseness=.85] (pb);
  \draw (xm) to[out=20,in=163,looseness=.8] (px);
  \draw (xu) to[out=20,in=163,looseness=.8] (pxm);
}}

\newcommand\pfycpinline{\fibration{
  \symb{(}
  \fibre{axbase}{\allx}{
     \over{black vx}{\lo}{xl}
     \over{black vx}{\hi}{xu}
  }
  \symb{\mkern2mu}
  \fibre{pxbase}{\px}{
    \over{blue vx}{\lo}{pxl}
    \over{red vx}{\hi}{pxu}
  }
  \symb{\mkern2mu)\mkern4mu\implies\mkern4mu}
  \fibre{aybase}{\ally}{
    \over{black vx}{\lo}{y}
  }
  \symb{\mkern2mu(\mkern1mu}
  \fibre{pybase}{\py}{
    \over{blue vx}{\lo}{py}
  }
  \symb{\mkern-2mu\wedge\mkern-2mu}
  \fibre{pfybase}{\pfy}{
    \over{red vx}{\lo}{pfy}
  }
  \symb{\mkern1mu)}
  \draw (xu) -- (pxu);
  \draw (xl) -- (pxl);
  \draw (py) -- (pfy);
}}

\newcommand\cponeinline{\eabcpinline}
\newcommand\cptwoinline{\pfycpinline}
\newcommand\cpthreeinline{\pabcpinline}
\newcommand\cpfourinline{\onioncpinline}

\newcommand\drinkersquare{
\lvx{-\halfedgelen,\halfedgelen}{x}{x}{135}
\lvx{\halfedgelen,\halfedgelen}{px}{\mkern.5mu\likex\ppx}{45}
\lvx{-\halfedgelen,-\halfedgelen}{y}{\mkern.5mu\likex y}{-135}
\lvx{\halfedgelen,-\halfedgelen}{py}{\py}{-45}
}
\newcommand\drinkersquareedges{
  \e x y
  \e x {px}
  \e x {py}
}

\newcommand\drinkersquarewithedges{\drinkersquare\drinkersquareedges}

\newcommand\drinkergraphpic{\gr{\drinkersquarewithedges}}

\newcommand\drinkercotreepic{\tr{\tikzset{level 1/.style={sibling distance=8ex},level 2/.style={sibling distance=4ex}}
  \joinroot
    child { \leaflab x }
    child { 
      \unionnode
      child { \leaflab y }
      child { \leaflab \ppx }
      child { \leaflab \py }
    }
  ;
}}

\newcommand\drinkerbindinggraphpic{\gr{\drinkersquare \de x {px} \de y {py}}}

\newcommand\liftingdiagrams{%
\begin{center}\begin{tikzpicture}[graph,outer sep=1.5pt,inner sep=0pt]\begin{math}
  \newcommand\rad{.8}
  \newcommand\vht{1.55}
  \newcommand\shortrad{.4}
  \newcommand\commonsquarepart{
    \node (liftw) at (-\rad,\vht) {$\skewlifting\vertexa$};
    \node[left=-1pt of liftw,outer sep=0,inner sep=0] (uniqueliftw) {$\exists\mkern1mu!$};
    \node (w) at (-\rad,0) {$\vertexa$};
    \node (v) at (\rad,1.55) {$v$};
    \node (fv) at (\rad,0) {$\fib(v)$};
    \begin{scope}[fibres]
      \fe{liftw}{w}
      \fe{v}{fv}
    \end{scope}
  }
  \begin{scope}[shift={(-3,0)}]    
    \commonsquarepart
    \eseps {w} {fv} {1.7pt} {1.7pt}
    \eseps {liftw} {v} {1.7pt} {1.7pt}
  \end{scope}
  \begin{scope}
    \node (v) at (\rad,1.55) {$v$};
    \node (fv) at (\rad,0) {$\fib(v)$};
    \node (fliftw) at (-\rad,.3){$\fib(\skewlifting w)$};
    \node (liftw) at (-\rad,1.8){$\skewlifting\vertexa$};
    \node[left=-1pt of liftw,outer sep=0,inner sep=0] (uniqueliftw) {$\exists$};
    \node (w) at (-\shortrad,-.35) {$w$};
    \eseps {w} {fv} {1.7pt} {1.7pt}
    \eseps {liftw} {v} {1.7pt} {1.7pt}
    \eseps {fliftw} {fv} {.3pt} {1.7pt}
  \end{scope}
  \begin{scope}[fibres]
    \fe{liftw}{fliftw}
    \fe{v}{fv}
  \end{scope}
  \begin{scope}[shift={(3,0)}]    
    \commonsquarepart
    \deseps {w} {fv} {1.8pt} {1.5pt}
    \deseps {liftw} {v} {1.8pt} {1.5pt}
  \end{scope}
\end{math}\end{tikzpicture}\end{center}}

\tikzset{
    position/.style args={#1:#2 from #3}{
        at=(#3.#1), anchor=#1+180, shift=(#1:#2)
    }
}
\newcommand\labsep[4]{
\node[outer sep=0,inner sep=0,position=#3:#4 from #1] {\ensuremath{#2}};}

\newcommand\coverylabelangle{-35}
\newcommand\coverppxlabelangle{0}
\newcommand\coverppxlabeldist{2pt}
\newcommand\addalldrinkercoverlabels[3]{
  \labsep{xa}{#1}{170}{2pt}
  \labsep{xb}{#2}{-170}{2pt}
  \labsep{ya}{\likex #3}{\coverylabelangle}{2pt}
  \labsep{ppxb}{\likex{\mkern-1mu\pp\mkern-.6mu #2}}{\coverppxlabelangle}{\coverppxlabeldist}
  \labsep{pya}{\likex{p #3}}{0}{2pt}
}
\newcommand\adddrinkercoverlabels[2]{\addalldrinkercoverlabels{#1}{#2}{y}}

\newcommand\drinkercoververtices{
  \begin{scope}[shift={(0,\drinkercovergap)}]
    \vx{\drinkerx,0}{xa}
    \vx{\drinkerxxx,-\drinkercoverseconddown}{ya}
    \vx{\drinkerxxxx,0}{pya}
  \end{scope}
  \vx{\drinkerx,0}{xb}
  \vx{\drinkerxx,-\drinkercoverfirstdown}{ppxb}
}

\newcommand\drinkercover{
  \drinkercoververtices
  \drinkercoveredges
}

\newcommand\drinkerfib{
\begin{scope}[shift={(0,\fibheight)}]
  \drinkercover
\end{scope}
\drinkerbase
\drinkerfibres
}

\newcommand\drinkerfiblabelledpair[2]{
  \drinkerfib
  \adddrinkercoverlabels{#1}{#2}
}

\newcommand\drinkerfibvertices{
  \drinkerbasevertices
  \begin{scope}[shift={(0,\fibheight)}]\drinkercoververtices\end{scope}
}

\newcommand\drinkerfibverticesandfibres{
  \drinkerfibvertices
  \drinkerfibres
}

\newcommand\drinkerbindingfiblabelled[2]{
  \drinkerfibverticesandfibres
  \adddrinkercoverlabels{#1}{#2}
  \de{ya}{pya}
  \de{xb}{ppxb}
  \de{y}{py}
  \de{x}{ppx}
}

\newcommand\twolinkassignment{\sublist{\subst x z,\subst y{fz}}}

\newcommand\twolinkfographvertices{
  \begin{scope}[shift={(-\halfedgelen,0)}]
    \lvx{-\edgelen,\halfedgelen}{x}{x}{180}
    \redlvx{0,\halfedgelen}{ppx}{\ppx}{90}
    \lvx{-\edgelen,-\halfedgelen}{y}{y}{180}
    \bluelvxsep{0,-\halfedgelen}{qqy}{\likex\qqy}{-90}{5pt}
  \end{scope}
  \begin{scope}[shift={(\halfedgelen,0)}]
    \redlvx{0,\halfedgelen}{pz}{\pz}{90}
    \bluelvxsep{0,-\halfedgelen}{qfz}{\likex{\qfz}}{-90}{5pt}
  \end{scope}
  \begin{scope}[shift={(\edgelen,0)}]
     \lvx{\halfedgelen,0}zz0
  \end{scope}
}
\newcommand\twolinkfographedges{
  \e x {ppx}
  \e y {qqy}
  \e {pz} {qfz}
}
\newcommand\twolinkfograph{\twolinkfographvertices\twolinkfographedges}

\newcommand\twolinkleapgraphedges{
  \begin{scope}[leap]
    \draw (ppx) to[out=-30,in=-150] (pz);
    \draw (qqy) to[out=30,in=150] (qfz);
    \draw (x) to[out=65,in=95,looseness=1.5] (z);
    \draw (y) to[out=-65,in=-95,looseness=1.5] (z);
  \end{scope}
}
\newcommand\twolinkleapgraph{
  \twolinkfographvertices
  \twolinkleapgraphedges
}

\newcommand\peirceformula{((\cneg{p} \vlor q) \vlan \cneg{p}) \vlor p}

\newcommand\peirceimpliesformula{\left(\rule{0ex}{1.7ex}\mkern-1mu(p\mkern-1mu\tightimplies\mkern-2mu q)\mkern-1mu\tightimplies p\right)\mkern-2mu\tightimplies\mkern-.3mu p}

\newcommand\peircexone{-.9}
\newcommand\peircextwo{-.3}
\newcommand\peircexthree{.35}
\newcommand\peircexfour{1.05}
\newcommand\peircecoverradius{.22}
\newcommand\peircey{-.14}
\newcommand\peircecp{
  \begin{scope}[shift={(0,\peircefibheight)}]
    \begin{scope}
      \begin{scope}
        \redvx{\peircexone,0}{cppl}
        \bluevx{\peircexthree,0}{cppr}
      \end{scope}
      \begin{scope}[shift={(0,-\peircecoverradius)}]
        \bluevx{\peircexfour,0}{cpd}
      \end{scope}
      \begin{scope}[shift={(0,\peircecoverradius)}]
        \redvx{\peircexfour,0}{cpu}
      \end{scope}
    \end{scope}
  \end{scope}
  \lvxd{\peircexone,0} {ppl} {\pp}
  \lvxd{\peircextwo,\peircey}{q}{q}
  \lvxd{\peircexthree,0}{ppr}{\pp}
  \lvxd{\peircexfour,0}{p}{p}
  \e {cppl} {cppr}
  \e {ppl} {ppr}
  \e {q} {ppr}
  \begin{scope}[fibres]
    \fe {cppl}{ppl}
    \fe {cppr}{ppr}
    \fe {cpu}{cpd}
    \fe {cpd}{p}
  \end{scope}
}
\newcommand\peircecppic{
  \begin{tikzpicture}
    \peircecp
  \end{tikzpicture}
}

\begin{document}
\vlnosmallleftlabels%
\bstctlcite{refs:BSTcontrol}%
%
\title{Combinatorial Proofs and Decomposition Theorems for First-order Logic}


\author{
  \IEEEauthorblockN{Dominic J.\ D.\ Hughes}
  \IEEEauthorblockA{%
	Logic Group\\
	U.C.\ Berkeley\\
	USA}
\and
\IEEEauthorblockN{Lutz Straßburger}
\IEEEauthorblockA{Inria, Equipe Partout\\
Ecole Polytechnique, LIX \\ 
France}
\and
\IEEEauthorblockN{Jui-Hsuan Wu}
\IEEEauthorblockA{Ecole Normale Supérieure\\
France}}


%




\IEEEoverridecommandlockouts
\IEEEpubid{\makebox[\columnwidth]{\emph{[Long version of the LICS 2021 paper, with full proofs in the appendix.]}\hfill}
    \hspace{\columnsep}\makebox[\columnwidth]{ }}

\maketitle

\begin{abstract}
  We uncover a close relationship between combinatorial and
  syntactic proofs for first-order logic (without equality). Whereas syntactic proofs are
  formalized in a deductive proof system based on inference rules,
  a combinatorial proof is a syntax-free presentation of a proof
  that is independent from any set of inference rules. We show that
  the two proof representations are related via a deep inference
  decomposition theorem that establishes a new kind of normal form
  for syntactic proofs. This yields (a) a simple proof
  of soundness and completeness for first-order combinatorial
  proofs, and (b) a full completeness theorem: every combinatorial
  proof is the image of a syntactic proof.
\end{abstract}


%



\section{Introduction}

First-order predicate logic is a cornerstone of modern
logic. Since its formalisation by Frege~\cite{frege:79} it has seen a
growing usage in many fields of mathematics and computer science. Upon
the development of proof theory by Hilbert~\cite{hilbert:22},
\emph{proofs} became first-class citizens as mathematical objects that
could be studied on their own. Since Gentzen's \emph{sequent
calculus}~\cite{gentzen:35:I,gentzen:35:II}, many other proof systems
have been developed that allow the implementation of efficient proof
search, for example \emph{analytic tableaux}~\cite{smullyan:68} or
\emph{resolution}~\cite{robinson:65}. Despite the immense progress
made in proof theory in general and in the area of
automated and interactive theorem provers in
particular, we still have
no satisfactory notion of proof identity for first-order logic. In
this respect, proof theory is quite different from any other
mathematical field. For example in group theory, two groups are
\emph{the same} iff they are isomorphic; in topology, two spaces are
\emph{the same} iff they are homeomorphic; etc. In proof theory, we
have no such notion telling us when two proofs are \emph{the same},
even though Hilbert was considering this problem as a possible 24th
problem~\cite{thiele:03} for his famous lecture in
1900~\cite{hilbert:00}, before proof theory existed as a mathematical
field.

The main reason for this problem is that formal proofs, as they are
usually studied in logic, are inextricably tied to the syntactic
(inference rule based) proof system in which they are carried out. And
it is difficult to compare two proofs that are produced within two
different syntactic proof systems, based on different sets of inference rules.
Consider the derivations in Figure~\ref{fig:translate}, showing two proofs of the
formula $\peirceformula$ and two proofs of the formula $\veedrinkerformula$,
in sequent calculus (top) and in a deep inference system (bottom).
It is, \emph{a priori}, not clear how to compare them.%
\newcommand\lkpeirceproof{
\vlderivation{
  \vlin{\vlor}{}{\sqn{\peirceformula}}{
    \vlin{\conr}{}{\sqn{(\cneg{p} \vlor q) \vlan \cneg{p}, p}}{
      \vliin{\vlan}{}{\sqn{(\cneg{p} \vlor q) \vlan \cneg{p}, p, p}}{
        \vlin{\vlor}{}{\sqn{\cneg{p} \vlor q, p}}{
          \vlin{\weakr}{}{\sqn{\cneg{p}, q, p}}{
            \vlin{\axr}{}{\sqn{\cneg{p}, p}}{
              \vlhy{}}}}}{
        \vlin{\axr}{}{\sqn{\cneg{p}, p}}{
          \vlhy{}}}}}}
}%
\newcommand\dipeirceproof{
\vlderivation{
  \vlin{\acD}{}{\peirceformula}{
    \vlin{\fequ}{}{((\cneg{p} \vlor q) \vlan \cneg{p}) \vlor (p \vlor p)}{
      \vlin{\sw}{}{(\cneg{p} \vlan (\cneg{p} \vlor q)) \vlor p) \vlor p}{
        \vlin{\fequ}{}{(\cneg{p} \vlan ((\cneg{p} \vlor q) \vlor p)) \vlor p)}{
          \vlin{\sw}{}{(p \vlor (\cneg{p} \vlor q)) \vlan \cneg{p}) \vlor p}{
            \vlin{\aiD}{}{(p \vlor (\cneg{p} \vlor q)) \vlan (\cneg{p} \vlor p)}{
              \vlin{\tttD}{}{(p \vlor (\cneg{p} \vlor q)) \vlan \ttt}{
                \vlin{\fequ}{}{p \vlor (\cneg{p} \vlor q)}{
                  \vlin{\wrD}{}{(p \vlor \cneg{p}) \vlor q}{
                    \vlin{\aiD}{}{p \vlor \cneg{p}}{
                      \vlhy{\ttt}}}}}}}}}}}}
}%
\newcommand\lkdrinkerproof{
\vlderivation{
  \vlin{\conr}{}{\sqn{\exists x.(\cneg{p}x \vlor (\forall y.py))}}{
    \vlin{\exists}{}{\sqn{\exists x.(\cneg{p}x \vlor (\forall y.py)),\exists
x.(\cneg{p}x \vlor (\forall y.py))}}{
      \vlin{\vlor}{}{\sqn{\cneg{p}w \vlor (\forall y.py), \exists
x.(\cneg{p}x \vlor (\forall y.py))}}{
        \vlin{\forall}{}{\sqn{\cneg{p}w, \forall y.py, \exists
x.(\cneg{p}x \vlor (\forall y.py))}}{
          \vlin{\exists}{}{\sqn{\cneg{p}w, pz, \exists
x.(\cneg{p}x \vlor (\forall y.py))}}{
            \vlin{\vlor}{}{\sqn{\cneg{p}w, pz, \cneg{p}z \vlor (\forall y.py)}}{
              \vlin{\weakr}{}{\sqn{\cneg{p}w, pz, \cneg{p}z, \forall y.py}}{
                \vlin{\weakr}{}{\sqn{\cneg{p}w, pz, \cneg{p}z}}{
                \vlin{\axr}{}{\sqn{pz, \cneg{p}z}}{
                  \vlhy{}}}}}}}}}}}
}%
\newcommand\ksonedrinkerproof{
\vlderivation{
  \vlin{\acD}{}{\exists x.(\cneg{p}x \vlor (\forall y.py))}{
    \vlin{\mfaD}{}{\exists x.(\cneg{p}x \vlor (\forall y.(py \vlor py)))}{
      \vlin{\acD}{}{\exists x.(\cneg{p}x \vlor ((\forall y.py) \vlor (\forall y.py)))}{
        \vlin{\fequ}{}{\exists x.((\cneg{p}x \vlor \cneg{p}x) \vlor ((\forall y.py) \vlor (\forall y.py)))}{
          \vlin{\mexD}{}{\exists x.((\cneg{p}x \vlor (\forall y.py)) \vlor (\cneg{p}x \vlor (\forall y.py)))}{
            \vlin{\exists}{}{(\exists x.(\cneg{p}x \vlor (\forall y.py))) \vlor (\exists x.(\cneg{p}x \vlor (\forall y.py)))}{
              \vlin{\fequ}{}{(\cneg{p}w \vlor (\forall y.py)) \vlor (\exists x.(\cneg{p}x \vlor (\forall y.py)))}{
                \vlin{\exists}{}{\forall y.((\cneg{p}w \vlor py) \vlor (\exists x.(\cneg{p}x \vlor (\forall y.py))))}{
                  \vlin{\fequ}{}{\forall y.((\cneg{p}w \vlor py) \vlor (\cneg{p}y \vlor (\forall y.py)))}{
                    \vlin{\wD}{}{\forall y.((py \vlor \cneg{p}y) \vlor (\cneg{p}w \vlor (\forall y.py)))}{
                      \vlin{\aiD}{}{\forall y.(py \vlor \cneg{p}y)}{
                        \vlin{\forall}{}{\forall y.\ttt}{
                          \vlhy{\ttt}}}}}}}}}}}}}}
}%
\begin{figure}
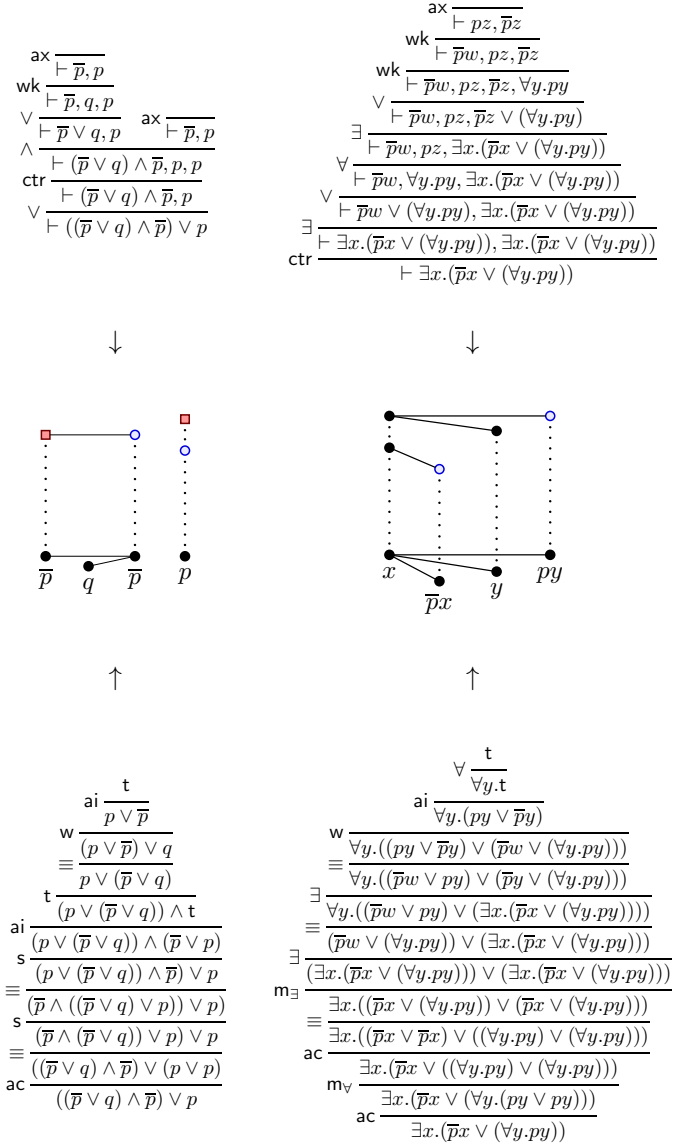
\vspace{-4ex}\begin{center}\hspace*{-0ex}\begin{math}
\newcommand\peircescale{.78}
\newcommand\drinkerscale{\peircescale}
\newcommand\cpscale{.95}
\begin{array}{@{}c@{\hspace{4ex}}c@{}}
{}\scalebox{\peircescale}{$\lkpeirceproof$}
&
{}\scalebox{\drinkerscale}{$\lkdrinkerproof$}
\\[15ex]
\downarrow
&
\downarrow
\\[3ex]
\vctr{\scalebox{\cpscale}{\peircecppic}}
&
\vctr{\raisebox{-2ex}[16.5ex]{\scalebox{\cpscale}{\drinkerfibcolouredpic}}}
\\[12ex]
\uparrow
&
\uparrow
\\[4ex]
{}\scalebox{\peircescale}{$\dipeirceproof$}
&
{}\scalebox{\drinkerscale}{$\ksonedrinkerproof$}
\end{array}
\end{math}\hspace*{-8ex}\end{center}\caption{Left:
syntactic proofs in sequent calculus (above)
and the calculus of structures (below)
which translate to the same propositional combinatorial proof (centre).
Right:
syntactic proofs 
in sequent calculus (above)
and the new calculus $\FOKS$ introduced in this paper (below),
which translate to the same first-order combinatorial proof (centre).%
}\label{fig:translate}\vspace{-4ex}\end{figure}%

This is where \emph{combinatorial proofs} come in. They were
introduced by Hughes~\cite{hughes:pws} for classical propositional
logic as a syntax-free notion of proof, and as a potential solution
to Hilbert's 24th problem~\cite{hughes:invar} (see
also~\cite{str:hilbert:24}). The basic idea is to abstract away from
the syntax of the inference rules used in inductively-generated proofs
and consider the proof as a combinatorial object, more precisely
as a special kind of
graph homomorphism. For example, a propositional combinatorial
proof of Peirce's law $\peirceimpliesformula=\peirceformula$ is shown mid-left in
Fig.\,\ref{fig:translate},
a homomorphism from a 4-vertex graph with two colours (above) to a 4-vertex graph labelled with
propositional variables (below); dotted vertical lines define the homomorphism, from above to below.

Several authors have illustrated how syntactic proofs in various
proof systems can be translated to propositional combinatorial proofs:
for sequent proofs in~\cite{hughes:invar}, for deep inference proofs
in~\cite{str:fscd17}, for Frege systems in~\cite{str:RR-9048}, and for
tableaux systems and resolution in~\cite{acc:str:18}. This enables a
natural definition of proof identity for propositional logic: two
proofs are \emph{the same} if they are mapped to the same
combinatorial proof.
For example, the left side of Fig.\,\ref{fig:translate} translates
syntactic proofs from sequent calculus and the calculus of structures into
the same combinatorial proofs, witnessing that the two syntactic proofs,
from different systems, are \emph{the same}.

Recently, Acclavio and Stra\ss burger extended this notion to relevant
logics~\cite{acc:str:relevant} and to modal
logics~\cite{acc:str:modal}, and Heijlties, Hughes and Stra\ss burger
have provided combinatorial proofs for intuitionistic propositional
logic~\cite{HHS:lics19}.

In this paper we advance the idea that combinatorial proofs can provide
a notion of proof identity for first-order logic. \emph{First-order
combinatorial proofs} were introduced by Hughes in~\cite{hughes:fopws}.
For example, a first-order combinatorial proof of Smullyan's 
\emph{drinker paradox} \mbox{$\drinkerformula$} $=$ \mbox{$\veedrinkerformula$}
is shown on the right of Fig.\,\ref{fig:translate}, 
a homomorphism from a 5-vertex partially coloured graph (with one colour) to a 4-vertex labelled graph.
However, even though Hughes proves soundness and completeness, the 
proof is unsatisfactory: (1) the soundness argument is long, intricate
and cumbersome, and (2) the completeness proof does not 
allow a syntactic proof to be read back from a combinatorial proof, i.e.,
completeness is not
\emph{sequentializable} \cite{girard:87} nor
\emph{full}~\cite{abramsky:jagadeesan:94}. 
A fundamental problem is that not all combinatorial
proofs can be obtained as translations of sequent calculus proofs.

We solve these issues by moving to a deep inference
system. More precisely, we introduce a new proof system,
$\FOKS$, for first-order logic, that (a) reflects every combinatorial
proof, i.e., there is a surjection from $\FOKS$ proofs to
combinatorial proofs, (b) yields simpler proofs of
soundness and completeness for combinatorial proofs, and (c) admits
new decomposition theorems establishing a precise correspondence
between certain syntactic inference rules and certain combinatorial
notions.
The right of Fig.\,\ref{fig:translate}
illustrates surjection in (a), and since the syntactic proofs 
in the two systems translate to the same combinatorial proof,
they can be considered \emph{the same}.

In general, a \emph{decomposition theorem} provides normal forms of
proofs, separating subsets of inference rules of a proof system. A
prominent example of a decomposition theorem is Herbrand's
theorem~\cite{herbrand:phd}, which allows a separation between the
propositional part and the quantifier part in a first-order
proof~\cite{gentzen:35:II,brunnler:06:herbrand}. Through the advent of
deep inference, new kinds of proof decompositions became possible,
most notably the separation between the linear part of a proof and the
resource management of a proof. It has been shown by
Stra{\ss}burger~\cite{str:07:RTA} that a proof in classical
propositional logic can be decomposed into a proof of multiplicative
linear logic, followed by a proof consisting only of contractions and
weakenings (see also \cite[\S4]{hughes:invar}).
In this paper we show that the same is possible for
first-order logic.

Combinatorial proofs and deep inference can be seen as opposite ends
of a spectrum: whereas deep inference allows for a very fine
granularity of inference rules---one inference rule in a standard
formalism, like sequent calculus or semantic tableaux, is usually
simulated by a sequence of different deep inference
rules---combinatorial proofs have completely abolished the concept of
inference rule. And yet, there is a close relationship between the
two, realized through a decomposition theorem, as we establish in this
paper.

\medskip
\noindent
    {\bf Outline:} This paper has three parts. First,
    in Sections~\ref{sec:fologic}--\ref{sec:foks} we present the
    preliminaries on first-order logic, first-order graphs,
    first-order combinatorial proofs, and the first-order proof system
    $\FOKS$. Second, in Section~\ref{sec:main} we state the main
    results. And third, in
    Sections~\ref{sec:LK1-KS1}--\ref{sec:summary} we give their
    proofs.


\section{Preliminaries: First-order Logic}\label{sec:fologic}

\subsection{Terms and Formulas}

Fix pairwise disjoint countably infinite sets $\VAR=\set{x,y,z,\ldots}$ of variables,
$\FUN=\set{f,g,\ldots}$ of function symbols, and 
$\PRED=\set{p,q,\ldots}$ of predicate symbols. Each function symbol and
each predicate symbol has a finite arity. Each predicate symbol $p$ has a
\bfit{dual}
$\dual{p}$ with $\dual{\dual p}\mkern6mu{\neq}\mkern5mu\dual p$. The grammars below generate
the set $\TERM$ of \bfit{terms}, denoted by $s,t,u,\ldots$, the set
$\ATOM$ of \bfit{atoms}, denoted by $a,b,c,\ldots$, and the set
$\FORM$ of \bfit{formulas}, denoted by $A,B,C,\ldots$:
\begin{equation*}
  \begin{array}{r@{~}l}
    t &\coloneqq~ x \mid f(t_1,\dots,t_n)
    \\[.5ex]
    a &\coloneqq~ \ttt\mid\fff\mid p(t_1,\dots,t_n) \mid \dual p(t_1,\dots,t_n)
    \\[.5ex]
    A &\coloneqq~ a \mid A\vlan A \mid A\vlor A \mid \exists x.A \mid \forall x.A
\end{array}
\end{equation*}
where the arity of $f$ and $p$ is $n$.
For better readability we often omit parentheses and write $ft_1 \dots t_n$ or $pt_1 \dots t_n$.
We consider the truth constants
$\ttt$ (\emph{true}) and $\fff$
(\emph{false}) as additional atoms, and consider all formulas in negation
normal form, where \bfit{negation} $(\dual\cdot)$ is defined on
atoms and formulas via De Morgan's laws:
\begin{equation*}
  \begin{array}{c@{\hspace{6ex}}c}
  \dual\ttt =\fff &
  \dual{p(t_1,\ldots,t_n)} = \dual{p}(t_1,\ldots,t_n) \\[.5ex]
  \dual\fff = \ttt &
  \dual{\dual p(t_1,\ldots,t_n)} = p(t_1,\ldots,t_n) \\[.5ex]
  \dual{\exists x.A} = \forall x.\dual{A}
  &
  \dual{A \vlan B} = \dual{A} \ \vlor \ \dual{B} \\[.5ex]
  \dual{\forall x.A} = \exists x.\dual{A}
  &
  \dual{A \vlor B} = \dual{A} \ \vlan \ \dual{B}
  \end{array}
\end{equation*}
Note $\dual{\dual a} = a$.
We write $A\implies B$ as an abbreviation for $\dual A\vlor B$.

A formula is \bfit{rectified} if all bound variables are distinct from
one another and from all free variables. Every formula can be
transformed into a logically equivalent rectified form by
bound variable renaming, e.g.\
\mbox{$(px \vee \exists xqx) \wedge \exists x r$} $\mapsto$ \mbox{$(px \vee \exists yqy) \wedge \exists zrz$}.
If we consider formulas equivalent
modulo bound variable renaming ($\alpha$-conversion), the
rectified form of a formula $A$ is unique, and we denote it
by $\rectif A$.

A \bfit{substitution} is a function $\sigma\colon\VAR\to\TERM$ that is
the identity almost everywhere. We denote substitutions as
$\sigma=\sublist{\subst{x_1}{t_1},\ldots,\subst{x_n}{t_n}}$, where
$\sigma(x_i)=t_i$ for $i=1..n$ and $\sigma(x)=x$ for all
$x\notin\set{x_1,\ldots,x_n}$. 
Write $A\sigma$ for the formula
obtained from $A$ by applying $\sigma$, i.e., by simultaneously
replacing all occurrences of $x_i$ by $t_i$.
A \bfit{variable renaming} is a substitution $\rho$ with $\rho(x)\in\VAR$ for all variables $x$.


\subsection{Sequent Calculus $\FOLK$}

\bfit{Sequents}, denoted by $\Gamma,\Delta,\ldots$, are finite
multisets of formulas, written as lists, separated by comma. The
\bfit{corresponding formula} of a (non-empty) sequent $\Gamma=A_1,A_2,\ldots,A_n$
is the disjunction of its formulas: $\form\Gamma=A_1\vlor
A_2\vlor\cdots\vlor A_n$. A sequent is \bfit{rectified} iff its
corresponding formula is.

In this paper we use the sequent calculus $\FOLK$, shown in
Figure~\ref{fig:LK1}, which is a one-sided variant of Gentzen's
original calculus~\cite{gentzen:35:I} for first-order logic. To
simplify some technicalities later in this paper, we include
the $\mixr$ rule.
\begin{thm}
  $\FOLK$ is sound and complete for first-order logic.
\end{thm}
\noindent For a proof, see any standard textbook, e.g.~\cite{TS:00}.

The linear fragment of $\FOLK$, i.e., the fragment without the rules
$\conr$ (\emph{contraction}) and $\weakr$ (\emph{weakening}) defines
\emph{first-order multiplicative linear
logic}~\cite{girard:87,girard:88} \emph{with
mix}~\cite{fleury:retore:94,bellin:97}~(MLL1+mix). We denote that
system here with $\FOMLL$ (shown in Figure~\ref{fig:LK1} in the dashed
box).

We will use the cut elimination theorem. The \bfit{cut} rule is
\begin{equation}
  \vliinf{\cutr}{}{\sqn{\Gamma,\Delta}}{\sqn{\Gamma,A}}{\sqn{\dual A,\Delta}}
\end{equation}

\begin{thm}
  \label{thm:cutelim}
  If a sequent $\sqn\Gamma$ is provable in $\FOLKcut$ then it is also
  provable in $\FOLK$. Furthermore, if $\sqn\Gamma$ is provable in
  $\FOMLLcut$ then it is also provable in $\FOMLL$.
\end{thm}
\noindent As before, this is standard, see e.g.~\cite{TS:00} for a proof.


\section{Preliminaries: First-order Graphs}\label{sec:fographs}


\subsection{Graphs}

A \bfit{graph} $\gG=\tuple{\vG,\eG}$ is a pair where $\vG$ is a finite
set of \bfit{vertices} and $\eG$ is a finite set of \bfit{edges},
which are two-element subsets of $\vG$. We write $vw$ for an edge
$\set{v,w}$.

Let $\gG=\tuple{\vG,\eG}$ and $\gH=\tuple{\vH,\eH}$ be graphs such
that $\vG\cap\vH=\emptyset$. A \bfit{homomorphism}
$\phi\colon\gG\to\gH$ is a function $\phi\colon\vG\to\vH$ such that if
$vw\in\eG$ then $\phi(v)\phi(w)\in\eH$. The \bfit{union} $\gG +\gH$ is
the graph $\tuple{\vG \cup\vH,\eG\cup \eH}$ and the \bfit{join} $\gG
\times \gH$ is the graph $\tuple{\vG \cup\vH,\eG \cup \eH \cup \set{vw
\mid v \in \vG, w \in \vH}}$.  A graph $\gG$ is \bfit{disconnected}
if $\gG=\gG_1+\gG_2$ for two non-empty graphs $\gG_1,\gG_2$, otherwise
it is \bfit{connected}.

A graph $\gG$ is \bfit{labelled} in a set $L$ if each vertex $v\in\vG$
has an associated \bfit{label} $\labelof v\in L$.
A graph $\gG$ is (partially) \bfit{coloured} if it carries a partial
equivalence relation $\linkingof\gG$ on $\vG$;
each equivalence class is a \bfit{colour}.\footnote{In
\cite{hughes:pws} and \cite{hughes:fopws}
adjacent vertices must have distinct colours, following
the standard definition of colouring in graph
theory.
We choose to omit this condition here, as it is implied
by the preclusion of bimatchings in Def.\,\ref{def:fonet}.}
A \bfit{vertex renaming} of
$\gG=\tuple{\vG,\eG}$ along a bijection
$(\hat\cdot)\colon\vG\to\hat\vG$ is the graph
$\hat\gG=\tuple{\hat\vG,\set{\hat v\hat w\mid vw\in\eG}}$, with
colouring and/or labelling inherited (i.e., $\hat v\samecol\hat w$ if
$v\samecol w$, and $\ell(\hat v)=\ell(v)$). Following standard graph
theory, we identify graphs modulo vertex renaming.

A \bfit{directed graph} $\gG =\tuple{\vG,\eG}$ is a set $\vG$ of
\bfit{vertices} and a set $\eG\subseteq \vG\times\vG$ of \bfit{direct edges}.
A \bfit{directed graph homomorphism}
$\phi\colon\tuple{\vG,\eG}\to\tuple{\vH,\eH}$ is a function $\phi\colon\vG\to\vH$ such that if
$\pair{v,w}\in\eG$ then $\pair{\phi(v),\phi(w)}\in\eH$.


\begin{figure}[!t]
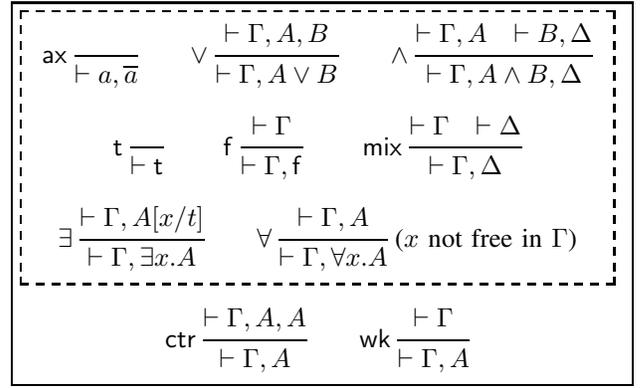

  \begin{center}
    \framebox{%
      $
      \begin{array}{@{}c@{}}
        \dbox{%
          $
          \begin{array}{c}
            \vlinf{\axr}{}{\sqn{a,\cneg a}}{}
            \qquad
            \vlinf{\vlor}{}{\sqn{\Gamma,A\vlor B}}{\sqn{\Gamma,A,B}}
            \qquad
            \vliinf{\vlan}{}{\sqn{\Gamma,A\vlan B, \Delta}}{\sqn{\Gamma,A}}{\sqn{B,\Delta}}
            \\ \\[-1ex]
            \vlinf{\ttt}{}{\sqn{\ttt}}{}
            \qquad\;
            \vlinf{\fff}{}{\sqn{\Gamma,\fff}}{\sqn{\Gamma}}
            \qquad\;
            \vliinf{\mixr}{}{\sqn{\Gamma,\Delta}}{\sqn\Gamma}{\sqn\Delta}
            \\ \\[-1ex]
            \vlinf{\exists}{}{\sqn{\Gamma,\exists x.A}}{\sqn{\Gamma,A\ssubst{x}{t}}}
            \qquad
            \vlinf{\forall}{\text{($x$ not free in $\Gamma$)}}{\sqn{\Gamma,\forall x.A}}{\sqn{\Gamma,A}}
          \end{array}
          $%
        }%
        \\\\[-1ex]
        \vlinf{\conr}{}{\sqn{\Gamma,A}}{\sqn{\Gamma,A,A}}
        \qquad
        \vlinf{\weakr}{}{\sqn{\Gamma,A}}{\sqn{\Gamma}}
      \end{array}
      $
    }
  \end{center}
  \caption{Sequent calculi $\FOLK$ (all rules) and $\FOMLL$ (rules in the dashed box)}
  \label{fig:LK1}
\end{figure}

\subsection{Cographs}

A graph $\gH=\tuple{\vH,\eH}$ is a \bfit{subgraph} of a graph
$\gG=\tuple{\vG,\eG}$ if $\vH\subseteq\vG$ and $\eH\subseteq\eG$. It
is \bfit{induced} if $v,w\in\vH$ and $vw\in\eG$ implies $vw\in\eH$. An induced
subgraph of $\gG=\tuple{\vG,\eG}$ is uniquely determined by its set of vertices
$V$ and we denote it by $\gG[V]$.
A graph is \bfit{$\gH$-free} if it does not contain $\gH$ as an induced
subgraph.  The graph $\Pfour$ is the (undirected) graph
$\tuple{\set{v_1, v_2, v_3, v_4},\set{v_1v_2, v_2v_3, v_3v_4}}$.  A
\bfit{cograph} is a $P_4$-free undirected graph. The interest in
cographs for our paper comes from the following well-known fact.

\begin{thm}[\cite{Ler71,CLS81}]\label{thm:cograph}
  A graph is a cograph iff it can be constructed from the singletons
  via the operations $+$ and~$\times$.
\end{thm}

In a graph $\gG$, the \bfit{neighbourhood} $N(v)$ of a vertex $v \in\vG$ is
$\set{w \mkern-2mu\mid\mkern-2mu vw \mkern-2mu\in\mkern-2mu\eG}$. A \bfit{module} is a set $M \mkern-2mu\subseteq\mkern-2mu \vG$
with $N(v) \mkern-1mu\setminus\mkern-1mu M = N(w) \mkern-1mu\setminus\mkern-1mu M$
for all $v, w \in M$. A module $M$ is \bfit{strong} if for every
module $M'$ we have $M' \mkern-2mu\subseteq\mkern-2mu M$, $M \mkern-2mu\subseteq\mkern-2mu M'$ or $M \cap M' =
\emptyset$. A module is \bfit{proper} if it has two or more vertices.


\subsection{Fographs}

A cograph is \bfit{logical} if every vertex is labelled by either an
atom or variable, and it has at least one atom-labelled vertex. An
atom-labelled vertex is a \bfit{literal} and a
variable-labelled vertex is a \bfit{binder}. A binder labelled
with $x$ is an \bfit{$x$-binder}. The \bfit{scope} of a binder
$b$ is the smallest proper strong module containing $b$.
An \bfit{$x$-literal} is a literal whose
atom contains the variable $x$. An $x$-binder \bfit{binds} every
$x$-literal in its scope.  In a logical cograph $\gG$, a binder $b$ is
\bfit{existential} (resp. \bfit{universal}) if, for every other vertex
$v$ in its scope, we have $bv \in \eG$ (resp. $bv \notin \eG$). An
$x$-binder is \bfit{legal} if its scope contains no other $x$-binder
and at least one literal.
\begin{definition}[{\cite[\S3]{hughes:fopws}}]
A \bfit{first-order graph} or \bfit{fograph} $\gG$ is a logical cograph whose binders are all legal. The
\bfit{binding graph} of $\gG$ is the directed graph $\bG=\tuple{\vG,
\set{(b, l) \mid b \text{ binds } l}}$.
\end{definition}
\noindent We define a mapping $\graphof\cdot$ from formulas to (labelled)
graphs, inductively as follows:
\begin{equation*}
  \begin{array}{r@{\;}l@{\hskip4em}r@{\;}l}
    \graphof{a} &= \single a \rlap{\hskip1em(for any atom $a$)}  \\[.9ex]
    \graphof{A \cor B} &= \graphof{A} + \graphof{B} &
    \graphof{\exists x . A} &= \single x \times \graphof{A} \\[.9ex]
    \graphof{A \cand B} &= \graphof{A} \times \graphof{B} &
    \graphof{\forall x . A} &= \single x + \graphof{A}    
  \end{array}
\end{equation*}
where we
write $\single\alpha$ for a single-vertex  labelled by $\alpha$.

\begin{example}
  Here is the fograph of the drinker formula $\drinkerformula=\veedrinkerformula$:
  \;\;
  
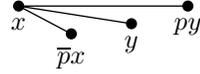
\hbox{\begin{tikzpicture}[baseline={([yshift={-1.5ex}]current bounding box.north)}]
      \drinkerbase
  \end{tikzpicture}}
\end{example}

\begin{lemma}
  If $A$ is a rectified formula then $\graphof A$ is a fograph.
\end{lemma}

\begin{proof}
  That $\graphof A$ is a logical cograph follows immediately from the definition and Theorem~\ref{thm:cograph}.
  The fact that every binder of $\graphof A$ is legal can be proved by
  structural induction on $A$.
\end{proof}
\begin{remark}
Note that $\graphof A$ need not be a fograph if $A$ is not
rectified. If $A = (\forall x. px) \vlor (\forall x. qx)$, then
$\graphof A = \single x \ \single px \ \single x \ \single qx$,
the scope of each $x$-binder contains all the vertices, in particular,
the other $x$-binder. On the other hand, there are non-rectified
formulas which are translated to fographs by $\graphof\cdot$. For
example, in the graph of $(\exists x. px) \vlor (\exists x. qx)$,
both $x$-binders are legal, as they are not in each other's
scope:
$\edgepairpic{x}{px}\mkern9mu\edgepairpic{x}{qx}$.
\end{remark}
\noindent We define a congruence relation $\fequ$ on formulas, called \bfit{equivalence},
by the following equations:
\begin{equation}
  \label{eq:fequ}
  \begin{array}{c@{\hskip1.5em}c}
    A \vlan B \fequ B \vlan A
    &
    (A \vlan B) \vlan C \fequ A \vlan (B \vlan C) 
    \\
    A \vlor B \fequ B \vlor A
    &
    (A \vlor B) \vlor C \fequ A \vlor (B \vlor C)
    \\
    \forall x . \forall y . A \fequ \forall y . \forall x . A
    &
    \forall x . (A \vlor B) \fequ (\forall x . A) \vlor B
    \\
    \exists x . \exists y . A \fequ \exists y . \exists x . A
    &
    \exists x . (A \vlan B) \fequ (\exists x . A) \vlan B
  \end{array}
\end{equation}
where $x$ must not be free in $B$ in the last two equations.
\begin{thm}[{\cite[\S10]{hughes:fopws}}]
  Let $A,B$ be rectified formulas. Then
  $$
  A\fequ B \iff \graphof A =\graphof B
  $$
\end{thm}

\begin{proof}
  A straightforward structural induction on formulas.
\end{proof}

\begin{example}
  \mbox{$\veedrinkerformula$} $\fequ$ \mbox{$\variantveedrinkerformula$}, and both formulas
  have the same (rectified) fograph
  $\drinkergraph$, below-left.
  \begin{equation*}
      \drinkergraphpic
      \hspace{8ex}
      \drinkercotreepic
      \hspace{8ex}
      \drinkerbindinggraphpic
  \end{equation*}
  Above-center we show the \emph{cotree} of the underlying
  cograph (illustrating the idea behind Theorem~\ref{thm:cograph}) and
  above-right is its binding graph $\bD$.
\end{example}


\section{First-order Combinatorial Proofs}\label{sec:focp}


\subsection{Fonets}

Two atoms are \bfit{pre-dual} if they are not $\ttt$ or $\fff$, and their predicate symbols are dual
(e.g.\ $p(x, y)$ and $\dual{p}(y, z)$) and two literals are \bfit{pre-dual} if their
labels (atoms) are pre-dual. A \bfit{linked fograph} $\tuple{\gC,\linkingof\gC}$ is a coloured fograph $\gC$ such
that every colour (i.e., equivalence class of~$\linkingof\gC$), called a \bfit{link}, consists of two pre-dual literals, and
every literal is either $\ttt$-labelled or in a link. Hence, in a linked fograph no vertex is labelled~$\fff$.

Let $\gC$ be a linked fograph. The set of links can be seen as a unification problem
by identifying dual predicate symbols. A \bfit{dualizer} of $\gC$ is a substitution $\delta$
unifying all the links of~$\gC$. Since a first-order unification problem is either
unsolvable or has a most general unifier, we can define the notion of \bfit{most
general dualizer}. A \bfit{dependency} is a pair $\set{\single x, \single y}$ of an
existential binder $\single x$ and a universal binder $\single y$ such that the most
general dualizer assigns to $x$ a term containing $y$. A \bfit{leap} is either a
link or a dependency. The \bfit{leap graph} $\lgC$ of $\gC$ is the undirected
graph $\tuple{\vC, \lpC}$ where $\lpC$ is the set of leaps of $\gC$.
A vertex set $W \subseteq \vC$ induces a \bfit{matching} in
$\gC$ if $W\neq\emptyset$ and 
for all $w \in W$, $N(w) \cap W$ is a singleton. We say that $W$ induces a
\bfit{bimatching} in $\gC$ if it induces a matching in $\gC$ and a matching in $\lgC$.

\begin{definition}[{\cite[\S5]{hughes:fopws}}]\label{def:fonet}
A \bfit{first-order net} or \bfit{fonet} is a linked fograph which
has a dualizer but no induced bimatching.  
\end{definition}
\begin{figure}[!t]
  \begin{center}
    \vspace{-3ex}
    \begin{tikzpicture}[graph]
      \begin{scope}[xshift=-2.2cm]\twolinkfograph\end{scope}
        \begin{scope}[xshift=2.2cm]\twolinkleapgraph\end{scope}
    \end{tikzpicture}\vspace{-4ex}
  \end{center}%
  \caption{\label{fig:leap}A fonet (left) with
    dualizer $\protect\twolinkassignment$
    and its
    leap graph (right).}
\end{figure}%
\noindent
Figure~\ref{fig:leap} shows a fonet with its dualizer and leap graph.%
%
%
\begin{figure*}%
  $$
  \begin{tikzpicture}\cptwoonformula
  \end{tikzpicture}
  \qquad
  \begin{tikzpicture}\cponeonformula
  \end{tikzpicture}
  \qquad
  \begin{tikzpicture}\cpthreeonformula
  \end{tikzpicture}
  \qquad
  \begin{tikzpicture}\cpfouronformula
  \end{tikzpicture}
  $$
  \vskip-2ex
  \caption{Four combinatorial proofs, each shown
    above the formula proved.  Here $x$ and $y$ are variables, $f$ is
    a unary function symbol, $a$ and $b$ are constants (nullary
    function symbols), $p$ is a unary predicate symbol, and $q$ is a
    binary predicate symbol. For each skew bifibration $\phi$,
    the variable substitution $\protect\rsubstof\phi$ is an identity,
    thus we can omit labels from each (coloured) source fograph
    (since the label of $v$ in the source is that of $\phi(v)$ in the target).}
  \label{fig:cps}
\end{figure*}
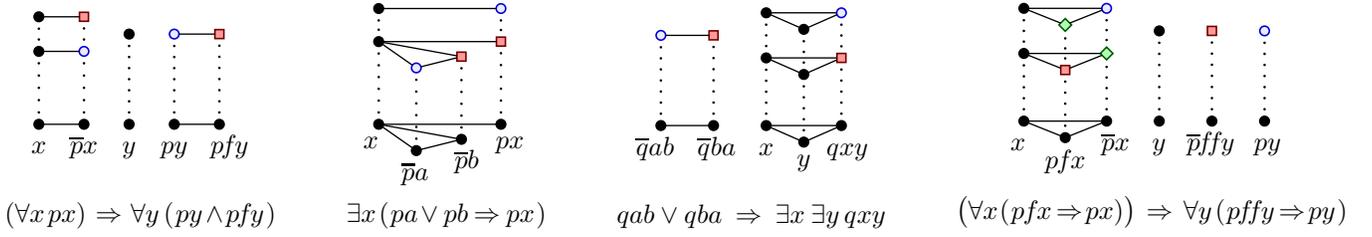%
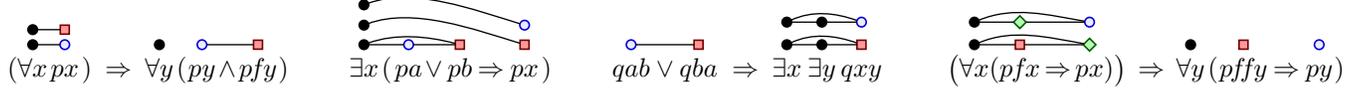
\begin{figure*}\vspace{-2ex}%
  $$
  \begin{tikzpicture}
    \cptwoinline
  \end{tikzpicture}
  \qquad\;
  \begin{tikzpicture}
    \cponeinline
  \end{tikzpicture}
  \qquad\;
  \begin{tikzpicture}
    \cpthreeinline
  \end{tikzpicture}
  \qquad\;
  \begin{tikzpicture}
    \cpfourinline
  \end{tikzpicture}
  $$
  \vskip-2ex
  \caption{Condensed forms of the four combinatorial proofs in Figure~\ref{fig:cps}. We do not show the lower graph, and indicate the mapping by the position of the vertices of the upper graph.}
  \label{fig:cps-condensed}
\end{figure*}%

\subsection{Skew Bifibrations}

A graph homomorphism $\phi\colon \tuple{\vG, \eG} \rightarrow
\tuple{\vH, \eH}$ is a \bfit{fibration} \cite{grothendieck:fibration,Gray:fibration} if
for all $v \in \vG$ and
$w\phi(v) \in \eH$, there exists a unique $\tilde{w} \in \vG$ such
that $\tilde{w}v \in \eG$ and $\phi(\tilde{w}) = w$ (indicated below-left),
and is a \bfit{skew fibration} \cite[\S3]{hughes:pws} if
for all $v \in \vG$ and $w\phi(v) \in \eH$
there exists $\tilde{w} \in \vG$ such that $\tilde{w}v \in \eG$ and
$\phi(\tilde{w})w \notin \eH$ (indicated below-centre).
A directed graph homomorphism is a
\bfit{fibration} if for all $v \in \vG$ and $\pair{w,\phi(v)} \in
\eH$, there exists a unique $\tilde{w} \in \vG$ such that $\pair{\tilde{w},v}
\in \eG$ and $\phi(\tilde{w}) = w$ (indicated below-right).
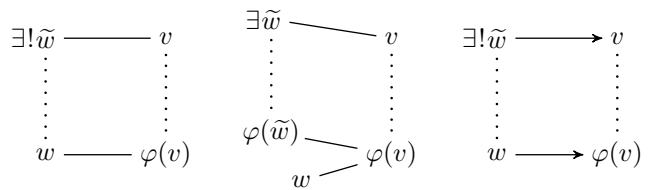
\liftingdiagrams
A \bfit{fograph homomorphism} $\phi=\tuple{\phi,\rsubstof\phi}$ is a
pair where $\phi\colon\gG\to\gH$ is a graph homomorphism between the
underlying graphs, and $\rsubstof\phi$, also called the
\bfit{substitution induced by $\phi$}, is a variable renaming
such that
for all $v\in\vG$ we have
$\labelof{\phi(v)}=\rsubstof\phi(\labelof v)$,
and $\rsubstof\phi$ is the identity on variables not in $\gG$.
Note that $\phi$ necessarily maps binders to binders
and literals to literals.
Since $\rsubstof\phi$ is fully determined by $\phi$ alone, we often
leave $\rsubstof\phi$ implicit.
A fograph homomorphism $\phi\colon\gG \rightarrow \gH$
\bfit{preserves existentials}
if for all existential binders $b$ in
$\gG$, the binder $\phi(b)$ is existential in $\gH$.
\begin{definition}[{\cite[\S4]{hughes:fopws}}]
  Let $\gG$ and $\gH$ be fographs. A \bfit{skew bifibration}
  $\phi\colon\gG\to\gH$ is an existential-preserving fograph
  homomorphism that is a skew fibration on $\tuple{\vG, \eG} \to
  \tuple{\vH, \eH}$ and a fibration on the binding graphs $\bG\to\bH$.
\end{definition}
\begin{example}
  Below-left is a skew bifibration, whose binding fibration
  is below-centre. When the labels on the source fograph can be inferred
  (modulo renaming),
  we often omit the labelling in the upper graph, as below-right.
  \begin{center}\begin{math}
  \scalebox{.9}{
  \begin{tikzpicture}[graph]\drinkerfiblabelledpair{x}{x}
  \end{tikzpicture}
  \quad
  \begin{tikzpicture}[graph]\drinkerbindingfiblabelled{x}{x}
  \end{tikzpicture}
  \quad
  \begin{tikzpicture}[graph]\drinkerfib
  \end{tikzpicture}}
\end{math}\end{center}\end{example}
\begin{definition}[{\cite[\S6]{hughes:fopws}}]
A \bfit{first-order combinatorial proof} (\bfit{FOCP}) of a fograph $\gG$ is a skew
bifibration $\phi\colon \gC \rightarrow \gG$ where $\gC$ is a fonet. A \bfit{first-order
combinatorial proof} of a formula $A$ is a combinatorial proof of its graph
$\graphof{A}$.
\end{definition}
\noindent Figure~\ref{fig:cps} shows examples of
FOCPs (taken from~\cite{hughes:fopws}), each above the formula it proves.
The same FOCPs are in
Figure~\ref{fig:cps-condensed} in \emph{condensed form}, with the formula
graph left implicit.

\begin{thm}[{\cite[\S6]{hughes:fopws}}]
  \label{thm:FOCP}
  FOCPs are sound and complete for first-order logic.
\end{thm}
\begin{remark}
  Our definition of FOCP is slightly more lax than the original
  definition of~\cite{hughes:fopws}, as we allow for a variable
  renaming $\rsubstof\phi$ which was restricted to be the identity
  in~\cite{hughes:fopws}.
\end{remark}


\section{First-order Deep Inference system $\FOKS$}\label{sec:foks}

\begin{figure}
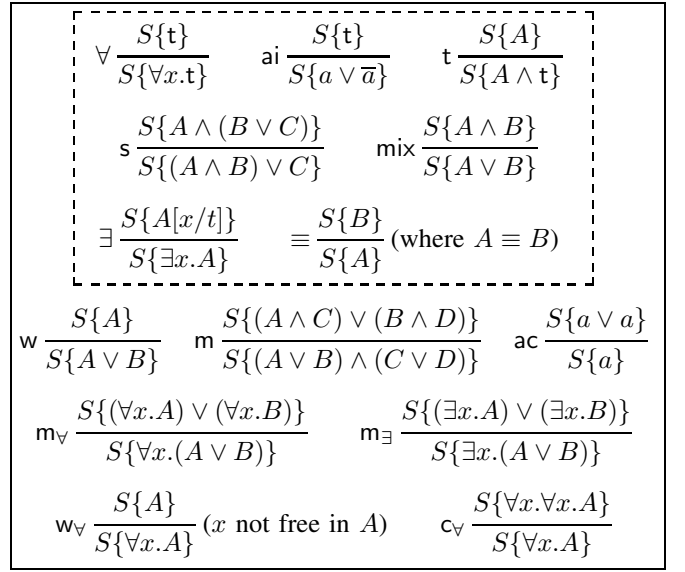

  \begin{center}
    \framebox{%
      $
      \begin{array}{@{}c@{}}
        \dbox{%
          $
          \begin{array}{c}
            \vlinf{\faD}{}{\Scons{\forall x.\ttt}}{\Scons\ttt}
            \qquad
            \vlinf{\aiD}{}{\Scons{a\vlor\cneg a}}{\Scons\ttt}
            \qquad
            \vlinf{\tttD}{}{\Scons{A\vlan\ttt}}{\Scons A}
            \\ \\[-1ex]
            \vlinf{\sw}{}{\Scons{(A\vlan B)\vlor C}}{\Scons{A\vlan (B\vlor C)}}
            \qquad
            \vlinf{\mix}{}{\Scons{A\vlor B}}{\Scons{A\vlan B}}            
            \\ \\[-1ex]
            \vlinf{\exD}{}{\Scons{\exists x.A}}{\Scons{A\ssubst x t}}
            \qquad
            \vlinf{\fequ}{\text{(where $A\fequ B$)}}{\Scons{A}}{\Scons{B}}
          \end{array}
          $
        }
        \\ \\[-1ex]
        \vlinf{\wrD}{}{\Scons{A\vlor B}}{\Scons{A}}
        \quad\;
        \vlinf{\me}{}{\Scons{(A\vlor B)\vlan(C\vlor D)}}{\Scons{(A\vlan C)\vlor(B\vlan D)}}
        \quad\;
        \vlinf{\acD}{}{\Scons{a}}{\Scons{a\vlor a}}
        \\ \\[-1ex]
        \vlinf{\mfaD}{}{\Scons{\forall x.(A\vlor B)}}{\Scons{(\forall x.A)\vlor(\forall x.B)}}
        \qquad
        \vlinf{\mexD}{}{\Scons{\exists x.(A\vlor B)}}{\Scons{(\exists x.A)\vlor(\exists x.B)}}        
        \\ \\[-1ex]
        \vlinf{\wfaD}{\text{($x$ not free in $A$)}}{\Scons{\forall x.A}}{\Scons{A}}
        \qquad
        \vlinf{\cfaD}{}{\Scons{\forall x. A}}{\Scons{\forall x.\forall x.A}}
      \end{array}
      $
    }
  \end{center}
  \caption{Deep inference systems $\FOKS$ (all rules) and $\FOMLS$ (rules in the dashed box)}
  \label{fig:KS1}
  \vskip-.2ex
\end{figure}
In contrast to standard proof formalisms, like sequent calculi or
tableaux, where inference rules decompose the principal formula along
its root connective, \emph{deep inference rules} apply like
rewriting rules inside any (positive) formula or sequent
\bfit{context}, which is denoted by $\Sconhole$, and which is a
formula (resp.~sequent) with exactly one occurrence of the \bfit{hole}
$\conhole$ in the position of an atom. Then $\Scons A$ is the result
of replacing the hole $\conhole$ in $\Sconhole$ with $A$.

Figure~\ref{fig:KS1} shows the inference rules for the deep inference
system~$\FOKS$ introduced in this paper. It is a
variation of the systems presented by Br\"unnler~\cite{brunnler:phd}
and Ralph~\cite{ralph:phd} in their PhD-theses. The main differences
are (i) the explicit presence of the $\mixr$-rule, (ii) a different
choice of how the formula equivalence $\fequ$ is defined, (iii) an
explicit rule for the equivalence, and (iv) new inference rules
$\wfaD$ and $\cfaD$. The reason behind these design choices is to
obtain the correspondence with combinatorial proofs and the full
completeness result.

We consider here only the cut-free fragment, as cut-elimination for
deep inference systems has already been discussed
elsewhere (e.g.~\cite{brunnler:06:herbrand,alertubella:guglielmi:18}).%
\footnote{In the deep
inference literature, the cut-free fragment is also called the
\emph{down-fragment}. But as we do not discuss the \emph{up-fragment}
here, we omit the down-arrows $\downarrow$ in the rule names.}
As with the sequent system $\FOLK$, we also need for $\FOKS$ the
\emph{linear fragment}, $\FOMLS$, and that is shown
in Figure~\ref{fig:KS1} in the dashed box.

We write $\vlder{\sysS}{\Deri}{A}{B}$ to denote a derivation $\Deri$
from $B$ to $A$ using the rules from system $\sysS$. A formula $A$ is
\bfit{provable} in a system $\sysS$ if there is a derivation in
$\sysS$ from $\ttt$ to $A$.

We will for some results also employ the
general (non-atomic) version of the contraction rule:
\begin{equation}
  \vlinf{\cD}{}{\Scons{A}}{\Scons{A\vlor A}}  
\end{equation}


\section{Main Results}
\label{sec:main}

We state the main results of this paper here,
and prove them in later sections. The
first is routine and expected, but must be proved nonetheless:

\begin{thm}\label{thm:KS1}
  $\FOKS$ is sound and complete for first-order logic.
\end{thm}

Our second result is more surprising, as it is a very strong
decomposition result for first-order logic.

\begin{thm}\label{thm:decomposition}
  For every derivation $\vlder{\FOKS}{\Deri}{A}{\ttt}$ there are $\fff$-free formulas $A_1,\ldots,A_5$ and a derivation
  \begin{equation*}
    \scalebox{.9}{$
    \vlderivation{
      \vlde{\set{\wrD,\wfaD,\fequ}}{}{A}{
        \vlde{\set{\acD,\cfaD}}{}{A_1}{
          \vlde{\set{\me,\mfaD,\mexD,\fequ}}{}{A_2}{
            \vlde{\set{\exD}}{}{A_3}{
              \vlde{\set{\sw,\mix,\fequ}}{}{A_4}{
                \vlde{\set{\faD,\aiD,\tttD}}{}{A_5}{
                  \vlhy{\ttt}}}}}}}}
    $}
  \end{equation*}
\end{thm}
\noindent
This theorem is stronger than the existing decompositions for
first-order logic, which either separate only atomic contraction and
atomic weakening~\cite{brunnler:phd} or only
contraction~\cite{ralph:phd} or only the quantifiers in form of a
Herbrand theorem~\cite{brunnler:06:locality,ralph:phd}.

Theorem~\ref{thm:decomposition} is also the reason why we have the rules
$\wfaD$ and $\cfaD$ in system $\FOKS$, as these rules are derivable
with the other rules. However, they are needed to obtain this
decomposition.
%
  Figure~\ref{fig:example-decompose} shows an example of a decomposed derivation in $\FOKS$ of the
  formula $(\exists x. \cneg{p}x) \vlor (\forall y.(py \vlan pfy))$.
%
\begin{figure}
  \begin{equation*}
    \scalebox{.9}{$
    \vlderivation{
      \vlin{\acD}{}{(\exists x. \cneg{p}x) \vlor (\forall y.(py \vlan pfy))}{
        \vlin{\mexD}{}{(\exists x. (\cneg{p}x \vlor \cneg{p}x) \vlor (\forall
          y.(py \vlan pfy))}{
          \vlin{\fequ}{}{((\exists x. \cneg{p}x) \vlor (\exists x.\cneg{p}x))
            \vlor (\forall y.(py \vlan pfy)}{
            \vlin{\exists}{}{\forall y.(((\exists x. \cneg{p}x) \vlor (\exists
              x.\cneg{p}x)) \vlor (py \vlan pfy))}{
              \vlin{\exists}{}{\forall y.((\cneg{p}y \vlor (\exists x.\cneg{p}x))
                \vlor (py \vlan pfy))}{
                \vlin{\fequ}{}{\forall y.((\cneg{p}y \vlor \cneg{p}fy) \vlor
                  (py \vlan pfy))}{
              \vlin{\sw}{}{\forall y.(\cneg{p}y \vlor ((py \vlan pfy)
                \vlor \cneg{p}fy))}{
                \vlin{\fequ}{}{\forall y.(\cneg{p}y \vlor (py \vlan (pfy
                  \vlor \cneg{p}fy)))}{
                  \vlin{\aiD}{}{\forall y.((\cneg{p}y \vlor py) \vlan (pfy
                    \vlor \cneg{p}fy))}{
                    \vlin{\aiD}{}{\forall y.((\cneg{p}y \vlor py) \vlan
                      \ttt)}{
                      \vlin{\ttt}{}{\forall y.(\ttt \vlan \ttt)}{
                        \vlin{\forall}{}{\forall y.\ttt}{
                          \vlhy{\ttt}}}}}}}}}}}}}}
    $}
  \end{equation*}    
    \caption{Example derivation in decomposed form of Theorem~\ref{thm:decomposition}}
    \label{fig:example-decompose}
  \end{figure}

A weaker version of Theorem~\ref{thm:decomposition} will
also be useful:
\begin{thm}\label{thm:decompositionA}
  For every derivation $\rule{0ex}{4.7ex}\upsmash{\vlder{\FOKS}{\Deri}{A}{\ttt}}$
there is a formula~$A'$ with no occurrence of $\fff$ and a derivation
  \begin{equation*}
    \vlderivation{
      \vlde{\set{\wrD,\cD,\fequ}}{}{A}{
        \vlde{\FOMLS}{}{A'}{
          \vlhy{\ttt}}}}
  \end{equation*}
\end{thm}
\noindent Here $A'$ corresponds to $A_3$ of Theorem~\ref{thm:decomposition}.

We now establish the connection between derivations in $\FOKS$ and
combinatorial proofs.

\begin{thm}\label{thm:CP-DI}
  Let  $\phi\colon\gC\to\gA$ be a combinatorial proof and let $A$ be a formula with $\gA=\graphof A$. Then there is a derivation
  \begin{equation}
    \label{eq:decom}
    \vlderivation{
      \vlde{\set{\wrD,\wfaD,\acD,\cfaD,\me,\mfaD,\mexD,\fequ}}{\Deri_2}{A}{
        \vlde{\FOMLS}{\Deri_1}{A'}{
          \vlhy{\ttt}}}}
  \end{equation}
  for some $A'\fequ C\rsubstof\phi$ where $C$ is a formula with $\graphof
  C=\gC$ and $\rsubstof\phi$ is the variable renaming substitution
  induced by~$\phi$.  Conversely,
  whenever we have a derivation as in~\eqref{eq:decom} above, such that $\fff$ does not occur in~$A'$, then
  there is a combinatorial proof $\phi\colon\gC\to\graphof A$ such
  that $\gC=\graphof{\rectif{A'}}$.
\end{thm}

Furthermore, in the proof of Theorem~\ref{thm:CP-DI}, we will see that
(i) the links in the fonet $\gC$ correspond precisely to the pairs of
atoms that meet in the instances of the $\aiD$-rule in the derivation
$\Deri_1$, and (ii) the ''flow-graph'' of $\Deri_2$ that traces the
quantifier- and atom-occurrences in the derivation corresponds exactly
to the vertex-mapping induced by $\phi$. To give an example, consider
the derivation in Figure~\ref{fig:example-decompose} which corresponds
to the left-most combinatorial proof in Figures~\ref{fig:cps}
and~\ref{fig:cps-condensed}.

Thus, combinatorial proofs are closely related to derivations of the
form~\eqref{eq:decom}, and since by Theorem~\ref{thm:decomposition}
every derivation can be transformed into that form, we can say that
combinatorial proofs provide a canonical proof representation for first-order
logic, similarly to what proof nets are for linear
logic~\cite{girard:96:PN}.

Finally, Theorems~\ref{thm:KS1}, \ref{thm:decomposition}
and~\ref{thm:CP-DI} imply Theorem~\ref{thm:FOCP}, which means that we
have here an alternative proof of the soundness and completeness for
first-order combinatorial proofs which is simpler than the one given
in~\cite{hughes:fopws}, and improves with completeness being full
(a surjection from syntactic KS1 proofs onto combinatorial proofs).


\section{Translating between $\FOLK$ and $\FOKS$}
\label{sec:LK1-KS1}

We prove Theorems~\ref{thm:KS1},
\ref{thm:decomposition}, and~\ref{thm:decompositionA}, mainly by
translating derivations to and from the sequent calculus, and by rule
permutation arguments.

\subsection{The Linear Fragments $\FOMLL$ and $\FOMLS$}

We show that $\FOMLL$ and $\FOMLS$ are equivalent.

\begin{lemma}\label{lem:MLL1->MLS1}
  If $\sqn\Gamma$ is provable in $\FOMLL$ then $\form\Gamma$ is provable in  $\FOMLS$.
\end{lemma}

\begin{proof}
  This is a straightforward induction on the proof of $\sqn\Gamma$ in
  $\FOMLL$, making a case analysis on the bottommost rule instance. We
  show here only the case of
  $\vlinf{\forall}{}{\sqn{\Delta,\forall x.A}}{\sqn{\Delta,A}}$ (all
  other cases are simpler or have been shown before,
  e.g.~\cite{brunnler:phd}): By induction hypothesis, there is a proof
  of $\form\Delta\vlor A$ in $\FOMLS$. We can prefix every line in
  that proof by $\forall x$ and then compose the following derivation:
  \begin{center}\vspace{-2ex}\scalebox{.9}{\begin{math}
    \vlderivation{
      \vlin{\fequ}{}{\form\Delta\vlor\forall x.A}{
        \vlde{\FOMLS}{}{\forall x.\form\Delta\vlor A}{
          \vlin{\forall}{}{\forall x.\ttt}{
            \vlhy{\ttt}}}}}
  \end{math}}\end{center}
  where we can apply the $\fequ$-rule because $x$ is not free in $\Delta$.
\end{proof}

\begin{lemma}\label{lem:shallow}
  Let $\vlinf{\rr}{}{\Scons B}{\Scons A}$ be an inference rule in
  $\FOMLS$. Then the sequent $\sqn{\cneg A,B}$ is provable in $\FOMLL$.
\end{lemma}
\begin{proof}
  A routine exercise.
\end{proof}

\begin{lemma}\label{lem:context}
  Let $A,B$ be formulas, and let $\Sconhole$ be a (positive)
  context. If $\sqn{\cneg A,B}$ is provable in $\FOMLL$, then so is
  $\sqn{\cneg{\Scons A},\Scons B}$.
\end{lemma}

\begin{proof}
  A straightforward induction on $\Sconhole$. (see e.g.~\cite{gug:str:01})
\end{proof}

\begin{lemma}\label{lem:MLS1->MLL1}
  If a formula $C$ is provable in $\FOMLS$ then $\sqn C$ is provable in  $\FOMLL$. 
\end{lemma}

\begin{proof}
  We proceed by induction on the number of inference steps in the
  proof of $C$ in $\FOMLS$. Consider the bottommost rule instance
  $\vlinf{\rr}{}{\Scons B}{\Scons A}$. By induction hypothesis we have
  a $\FOMLL$ proof $\Proof$ of $\sqn{\Scons A}$.
      By Lemmas~\ref{lem:shallow} and~\ref{lem:context}, we have a $\FOMLL$ proof of
      $\sqn{\cneg{\Scons A},\Scons B}$. We can compose them via
      \begin{equation*}
        \vliinf{\cutr}{}{\sqn{\Scons B}}{\sqn{\Scons A}}{\sqn{\cneg{\Scons
          A},\Scons B}}
      \end{equation*}
      and then apply Theorem~\ref{thm:cutelim}.
  \end{proof}

\subsection{Contraction and Weakening}

The first observation here is that Lemmas~\ref{lem:MLL1->MLS1}--\ref{lem:MLS1->MLL1} from above also hold for $\FOLK$ and $\FOKS$. We therefore immediately have:

\begin{thm}\label{thm:LK1-KS1}
  For every sequent $\Gamma$, we have that $\sqn\Gamma$ is provable in
  $\FOLK$ if and only if $\form\Gamma$ is provable in $\FOKS$.
\end{thm}

Then Theorem~\ref{thm:KS1} is an immediate consequence. Let us now
proceed with providing further lemmas that will be needed for the
other results.

\begin{lemma}
  \label{lem:ac}
  The $\cD$-rule is derivable in $\set{\acD,\me,\mfaD,\mexD,\fequ}$. 
\end{lemma}

\begin{proof}
  This can be shown by a straightforward induction on $A$ (for details, see e.g.~\cite{brunnler:phd}).
\end{proof}

\begin{lemma}
  \label{lem:me}
  $\wfaD,\cfaD,\me,\mfaD,\mexD$ are derivable in \hbox{$\set{\wrD,\cD,\fequ}$}. 
\end{lemma}

\begin{proof}
  We only show the cases for $\wfaD$ and $\cfaD$ (for the others see~\cite{brunnler:phd}):
  \vadjust{\vskip-2ex}
  \begin{equation}
    \label{eq:wfa-cfa}
  \qquad
  \vlderivation{
    \vlin{\cD}{}{\forall x.A}{
      \vlin{\fequ}{}{\forall x.(A \vlor A)}{
        \vlin{\wrD}{}{A \vlor (\forall x.A)}{
          \vlhy{A}}}}}
  \qquad
  \quad
  \vlderivation{
    \vlin{\cD}{}{\forall x.A}{
      \vlin{\equiv}{}{(\forall x.A) \vlor (\forall x.A)}{
        \vlin{\wrD}{}{\forall x.((\forall x.A) \vlor A)}{
          \vlhy{\forall x.\forall x.A}}}}}
  \end{equation}
  where in the first derivation, $x$ is not free in $A$, and in the second one not free in $\forall x.A$.
\end{proof}

\begin{lemma}\label{lem:cw-atomic}
  Let $A$ and $B$ be formulas. Then
  \begin{equation*}
    \vlderivation{
      \vlde{\set{\wrD,\cD,\fequ}}{}{B}{
        \vlhy{A}}}
    \qquad
    \iff
    \qquad
    \vlderivation{
      \vlde{\set{\wrD,\wfaD,\acD,\cfaD,\me,\mfaD,\mexD,\fequ}}{}{B}{
        \vlhy{A}}}
  \end{equation*}
\end{lemma}

\begin{proof}
  Immediately from Lemmas~\ref{lem:ac} and~\ref{lem:me}.
\end{proof}

\begin{remark}
  Observe that Lemma~\ref{lem:cw-atomic} would also hold with the rules $\wfaD$ and~$\cfaD$ removed.
\end{remark}
\subsection{Rule Permutations}

\begin{thm}\label{thm:LK1-decompose}
  Let $\Gamma$ be a sequent. If\/ $\sqn\Gamma$ is provable in $\FOLK$ (as depicted on the left below) then there is a sequent
  $\Gamma'$ not containing any $\fff$, such that there is a derivation as shown on the right below:\vadjust{\vskip-2ex}
  \begin{equation*}
    \vlderivation{
        \vltrl{}{\FOLK}{\Deri}{\sqns{\Gamma}}{
          \vlhy{\quad}}{
          \vlhy{}}{
          \vlhy{}}}
    \qquad
    \Longrightarrow
    \qquad
    \vlderivation{
      \vlde{\set{\wrD,\cD,\fequ}}{\Deri_2}{\sqns{\form{\Gamma}}}{
        \vltrl{}{\FOMLL}{\Deri_1}{\sqns{\form{\Gamma'}}}{
          \vlhy{\quad}}{
          \vlhy{}}{
          \vlhy{}}}}
  \end{equation*}
\end{thm}

\begin{proof}
  First, we can replace every instance of the $\fff$-rule in $\Deri$ by
  $\weakr$. Then the instances of $\weakr$ and $\conr$ are replaced by
  $\wD$ and $\cD$, which can then be permuted down.
  Details are in
  Appendix~\ref{app:LK1-decompose}.
\end{proof}

\begin{lemma}\label{lem:MLS1-decomposition}
  For every derivation
  $\vlderivation{
      \vlde{\FOMLS}{}{A}{
        \vlhy{\ttt}}}$
  there are formulas $A'$ and $A''$ such that 
  \vspace{-2ex}\begin{equation*}\hspace{-2ex}
    \vlderivation{
      \vlde{\set{\exists}}{}{A}{
        \vlde{\set{\sw,\mix,\fequ}}{}{A'}{
          \vlde{\set{\forall,\aiD,\tttD}}{}{A''}{
            \vlhy{\ttt}}}}}
  \end{equation*}
\end{lemma}

\begin{proof}
  First, observe that the $\exists$ rule can be permuted under all the
  other rules since $A\ssubst{x}{t}$ has the same structure as $A$ and
  none of the other rules has a premise of the form $\Scons{\exists
    x.A}$. It suffices now to prove that all rules in $\set{\forall,
    \aiD, \tttD}$ can be permuted over the rules in $\set{\sw, \mix,
    \fequ}$, which is straightforward. For $\sw$, $\mix$, and the $\fequ$-instances that
  do not involve the quantifiers, the details can be found
  in~\cite{dissvonlutz}. The $\fequ$-instances concerning the
  quantifiers are admissible if the \hbox{$\exists$-rule} is not present.
\end{proof}

\begin{lemma}\label{lem:cw-decomposition}
  For every derivation
  $\vlderivation{
      \vlde{\set{\wrD,\wfaD,\acD,\cfaD,\me,\mfaD,\mexD,\fequ}}{}{B}{
        \vlhy{A}}}$
  there are formulas $A'$ and $B'$ such that 
  \begin{equation*}\hspace{-3ex}
    \vlderivation{
      \vlde{\set{\wrD,\wfaD,\fequ}}{}{B}{
        \vlde{\set{\acD,\cfaD}}{}{B'}{
          \vlde{\set{\me,\mfaD,\mexD,\fequ}}{}{A'}{
            \vlhy{A}}}}}
  \end{equation*}
\end{lemma}

\begin{proof}
  Permute all $\wD$ and $\wfaD$ instances to the bottom of
  the derivation, then permute all $\cD$ and
  $\cfaD$ below $\set{\me,\mfaD,\mexD}$. This involves a tedious but
  routine case analysis. However, unlike most other rule
  permutations in this paper, this has not been done before in the deep
  inference literature. 
  For this reason, we give the full case
  analysis in Appendix~\ref{app:cw-decomposition}.
  This
  Lemma is the reason for the presence of the rules $\wfaD$ and
  $\cfaD$, as without them the permutation cases in~\eqref{eq:wfa-cfa}
  could not be resolved.
\end{proof}

We can now complete the proof of Theorems~\ref{thm:decomposition} and~\ref{thm:decompositionA}.

\begin{proof}[Proof of Theorem~\ref{thm:decompositionA}]
  Assume we have a proof of $A$ in $\FOKS$. By
  Theorem~\ref{thm:LK1-KS1} we have a proof of $\sqn A$ in $\FOLK$ to
  which we can apply Theorem~\ref{thm:LK1-decompose}. Finally, we
  apply Lemma~\ref{lem:MLL1->MLS1} to get the desired shape.
\end{proof}

\begin{proof}[Proof of Theorem~\ref{thm:decomposition}]
  Assume we have a proof of $A$ in $\FOKS$. We first apply
  Theorem~\ref{thm:decompositionA}, and then
  Lemma~\ref{lem:MLS1-decomposition} to the upper half and
  Lemmas~\ref{lem:cw-atomic} and~\ref{lem:cw-decomposition} to the lower half.
\end{proof}


\section{Fonets and Linear Proofs}
\label{sec:linear}


\subsection{From $\FOMLL$ Proofs to Fonets}

Let $\Pi$ be a $\FOMLL$ proof of a rectified sequent $\sqn\Gamma$ not containing $\fff$. We now show
how $\Pi$ is translated into a linked fograph
$\fographof\Pi=\tuple{\graphof\Gamma,\linkingof\Pi}$. We proceed
inductively, making a case analysis on the last rule in $\Pi$. At the
same time we are constructing a dualizer $\dsubstof\Pi$, so that in the
end we can conclude that $\fographof\Pi$ is in fact a fonet.
\begin{enumerate}
  \addtolength\itemsep{.7ex}
\item $\Pi$ is $\upsmash{\vlinf{\axr}{}{\sqn{a,\cneg a}}{}}$ : Then the only
  link is $\set{a, \dual{a}}$, and $\dsubstof\Pi$ is empty.
\item $\Pi$ is $\smash{\vlinf{\ttt}{}{\sqn{\ttt}}{}}$ : Then
  $\linkingof\Pi$ and $\dsubstof\Pi$ are both empty.
\item The last rule in $\Pi$ is $\vliinf{\mixr}{}{\sqn{\Gamma',\Gamma''}}{\sqn\Gamma'}{\sqn\Gamma''}$ :
  By induction hypothesis, we have proofs $\Pi'$ and $\Pi''$ of $\Gamma'$ and $\Gamma''$, respectively. We have
  $\graphof{\Gamma}=\graphof{\Gamma'}+\graphof{\Gamma''}$ and we can let
  $\linkingof\Pi\;=\;\linkingof{\Pi'}\cup\linkingof{\Pi''}$ and
  $\dsubstof\Pi=\dsubstof{\Pi'}\cup\dsubstof{\Pi''}$.
\item The last rule in $\Pi$ is $\vlinf{\vlor}{}{\sqn{\Gamma_1,A\vlor
    B}}{\sqn{\Gamma_1,A,B}}$ : By induction hypothesis, there is
  a proof $\Pi'$ of $\Gamma'={\Gamma_1,A,B}$. We have
  $\graphof{\Gamma}=\graphof{\Gamma'}$ and let
  $\linkingof\Pi\;=\;\linkingof{\Pi'}$ and
  $\dsubstof\Pi=\dsubstof{\Pi'}$.
\item The last rule in $\Pi$ is $\vliinf{\vlan}{}{\sqn{\Gamma_1,A\vlan
    B, \Gamma_2}}{\sqn{\Gamma_1,A}}{\sqn{B,\Gamma_2}}$ : By induction
  hypothesis, we have proofs $\Pi'$ and $\Pi''$ of
  $\Gamma'=\Gamma_1,A$ and $\Gamma''=B,\Gamma_2$, respectively. We
  have $\graphof{\Gamma}=\graphof{\Gamma_1}+\mbox{$(\graphof A\times\graphof
  B)$}+\graphof{\Gamma_2}$ and we let
  $\linkingof\Pi\;=\;\linkingof{\Pi'}\cup\linkingof{\Pi''}$ and
  $\dsubstof\Pi=\dsubstof{\Pi'}\cup\dsubstof{\Pi''}$.
\item The last rule in $\Pi$ is
  $\vlinf{\exists}{}{\sqn{\Gamma_1,\exists x.A}}
  {\sqn{\Gamma_1,A\ssubst{x}{t}}}$ : By induction hypothesis, there is
  a proof $\Pi'$ of $\Gamma'={\Gamma_1,A\ssubst{x}{t}}$.  For each atom in
  $\Gamma'=\Gamma_1, A \ssubst{x}{t}$, there is a corresponding atom
  in $\Gamma=\Gamma_1, \exists x.A$. We can therefore define the
  linking $\linkingof{\Pi}$ from the linking $\linkingof{\Pi'}$ via
  this correspondence. Then, we let $\dsubstof\Pi$ be
  $\dsubstof{\Pi'}+\ssubst{x}{t}$. Since $\Gamma$ is rectified $x$
  does not yet occur in $\dsubstof{\Pi'}$. Hence $\dsubstof\Pi$ is a
  dualizer of $\fographof\Pi$.
\item The last rule in $\Pi$ is $\vlinf{\forall}{\text{($x$ not free
    in $\Gamma_1$)}}{\sqn{\Gamma_1,\forall x.A}}{\sqn{\Gamma_1,A}}$ :
  By induction hypothesis, there is a proof $\Pi'$ of
  $\Gamma'={\Gamma_1,A}$, which has the same atoms as in
  $\Gamma=\Gamma_1, \forall x.A$.  Hence, we can let
  $\linkingof\Pi\;=\;\linkingof{\Pi'}$ and
  $\dsubstof\Pi=\dsubstof{\Pi'}$. 
\end{enumerate}

\begin{thm}
  \label{thm:MLL1->fonet}
  If $\Pi$ is a $\FOMLL$ proof of a rectified $\fff$-free sequent $\sqn\Gamma$,
  then $\fographof\Pi$ is a fonet and $\dsubstof\Pi$ a dualizer for it.
\end{thm}

\begin{proof}
  We must show that none of the operations above introduces a
  bimatching. For cases 1--6, this is immediate. For case 7, observe
  that there is a potential dependency from each existential binder in
  $\graphof{\Gamma'}$ to the new $x$-binder $\single x$ in
  $\graphof{\Gamma}$. However, observe that this $\single x$ vertex is
  not connected to any vertex in $\graphof{\Gamma'}$, and hence no
  such new dependency can be extended to a bimatching. That
  $\dsubstof\Pi$ is a dualizer for $\fographof\Pi$ follows immediately
  from the construction. Hence,  $\fographof\Pi$ is a fonet.
\end{proof}

\subsection{From $\FOMLS$ Proofs to Fonets}

There is a more direct path from a $\FOMLL$ proof $\Pi$ of a rectified
sequent $\Gamma$ to the linked fograph $\fographof\Pi$: take
the fograph $\graphof\Gamma$, and let the equivalence classes of
$\linkingof\Pi$ be all the atom pairs that meet in an instance of
$\ax$, and $\dsubstof\Pi$ comprises the substitutions
at the $\exists$-rules in $\Pi$.
We chose the more cumbersome path above because it gives us a
direct proof of Theorem~\ref{thm:MLL1->fonet}.
However, for translating $\FOMLS$ derivation into fonets, we employ
exactly that direct path.

In a derivation in $\FOMLS$
where the conclusion is rectified, every line is also rectified, as
the only rules involving bound variables are $\forall$ and $\exists$
which (upwards) both remove a binder.
Therefore, we can call such a derivation \bfit{rectified}, and
for a non-rectified  $\FOMLS$ derivation $\Deri$ we can
define its \bfit{rectification} $\rectif\Deri$ inductively, by rectifying each
line, proceeding step-wise
from conclusion to premise.\footnote{As for formulas,
the rectification of a derivation is unique up to renaming of bound
variables.}

A rectified derivation $\vlder{\FOMLS}{\Deri}{A}{\ttt}$ determines a
substitution which maps the existential bound variables occurring in
$A$ to the terms substituted for them in the instances of the
$\exists$-rule in $\Phi$. We denote this substitution by
$\dsubstof\Phi$ and call it the \bfit{dualizer} of
$\Deri$. Furthermore, every atom occurring in the conclusion $A$ must
be consumed by a unique instance of the rule $\aiD$ in $\Deri$. This
allows us to define a (partial) equivalence relation $\linkingof\Deri$
on the atom occurrences in $A$ by $a\linkingof\Deri b$ if $a$ and $b$
are consumed by the same instance of $\aiD$ in~$\Deri$. We call
$\linkingof\Deri$ the \bfit{linking} of~$\Deri$, and define
$\fographof\Deri=\tuple{\graphof{A},\linkingof{\Deri}}$.


\begin{thm}\label{thm:MLS1->fonet}
  Let $\vlder{\FOMLS}{\Deri}{A}{\ttt}$ be a rectified derivation where $A$ is $\fff$-free. Then
  $\fographof\Deri$ is a fonet and $\dsubstof{\Deri}$ a dualizer for it.
\end{thm}

To prove this theorem, we have to show that no inference rule in
$\FOMLS$ can introduce a bimatching. To simplify the argument, we
introduce the \bfit{frame}~\cite{hughes:unifn} of the linked fograph $\gC$,
which is a linked (propositional) cograph in which the dependencies
between the binders in $\gC$ are encoded as links. 

More formally, let $C$ be a formula with $\graphof C=\gC$, to which we
exhaustively apply the following
subformula rewriting steps, to obtain a sequent $\frameof C$:
\begin{enumerate}
\item {\bf Encode dependencies as fresh links.} For each dependency
  $\set{\single {x_i},\single {y_j}}$ in $\gC$, with corresponding
  subformulas $\exists x_i. A$ and $\forall y_j. B$ in $C$, we pick a
  fresh (nullary) predicate symbol $q_{i,j}$, and then replace $\exists
  x_i. A$ by $\dual q_{i,j} \cand \exists x_i. A$, and replace $\forall y_j. B$ by
  $q_{i,j} \cor \forall y_j. B$.
\item {\bf Erase quantifiers.} After step 1, remove all the
  quantifiers, i.e., replace $\exists x_i.A$ by $A$ and replace $\forall
  y_j.B$ by $B$ everywhere.
\item {\bf Simplify atoms.} After step 2, replace every predicate
  $pt_1 \ldots t_n$ (resp. $\dual{p}t_1 \ldots t_n$) with a nullary
  predicate symbol $p$ (resp.~$\dual p$)
\end{enumerate}
Then $\linkingof{\frameof C}$ consists of the pairs induced by
$\linkingof\gC$ and the new pairs $\set{q_{i,j},\dual q_{i,j}}$ introduced in
step~1 above.
We call $\frameof C$ the \bfit{frame} of $C$ and we define the \bfit{frame} of $\gC$, denoted
$\frameof\gC$, as
$\tuple{\graphof{\frameof C},\linkingof{\frameof C}}$.

\begin{lemma}\label{lem:frame}
  If a linked fograph $\gC$ has an induced bimatching then so does its frame $\frameof\gC$. 
\end{lemma}

\begin{proof}
  Immediately from the construction of the frame.
\end{proof}

\begin{proof}[Proof of Theorem~\ref{thm:MLS1->fonet}]
  From $\Deri$ we construct a derivation $\frameof\Deri$ of $\frameof
  A$ in the propositional fragment of $\FOMLS$, such that
  \hbox{$\fographof{\frameof\Deri}=\frameof{\fographof\Deri}$}.  The
  rules $\aiD,\tttD,\mix$ and $\sw$ are translated trivially, and for
  $\fequ$, it suffices to observe that the frame construction is
  invariant under $\fequ$. Finally, for the rules $\forall$ and
  $\exists$, proceed as follows. Every instance of $\forall$ is
  replaced by the derivation on the right below:\footnote{For better
  readability we omit superfluous parentheses, knowing that we always
  have $\fequ$ incorporating associativity and commutativity of
  $\vlan$ and $\vlor$.} \vadjust{\vskip-4ex}
  \begin{equation*}
    \scalebox{.95}{$
    \vlinf{\forall}{}{\Scons{\forall y_j.\ttt}}{\Scons{\ttt}}
    \;\leadsto
    \vlderivation{
      \vlde{\set{\sw,\fequ}}{\DDeri_2}{\Scons{q_{h_1,j}\vlor\cdots\vlor
          q_{h_n,j}\vlor(\dual q_{h_1,j}\vlan\cdots\vlan\dual
          q_{h_n,j}\vlan\ttt)}}{
        \vlde{\set{\aiD,\tttD}}{\DDeri_1}{
          \Scons{(q_{h_1,j}\vlor \dual q_{h_1,j})\vlan\cdots\vlan(q_{h_n,j}\vlor \dual q_{h_n,j})\vlan\ttt}}{
          \vlhy{\ttt}}}}
    $}
  \end{equation*}
  where $h_1,\ldots,h_n$ range over the indices of the existential
  binders dependent on that $y_j$. It is easy to see how $\DDeri_1$ is
  constructed. The construction of $\DDeri_2$, using $\sw$ and $\fequ$, is standard, see,
  e.g.~\cite{ATS:esslli2019,gug:str:01,brunnler:tiu:01,dissvonlutz}. Then, every
  occurrence of $\forall y_j.F$ is replaced by
  $q_{h_1,j}\vlor\cdots\vlor q_{h_n,j}\vlor (\dual
  q_{h_1,j}\vlan\cdots\vlan\dual q_{h_n,j}\vlan F)$ in the derivation
  below that $\forall$-instance. Now, observe that all instances of
  the $\exists$-rule introducing $x_i$ dependent on $y_j$ must occur
  below in the derivation (otherwise $\Deri$ would not be rectified).
  Now consider such an instance $\vlinf{\exists}{}{\Scons{\exists
      x_i.B}}{\Scons{B\ssubst {x_i} t}}$.  Its context $\Sconhole$
  must contain all the $\forall y_j$ the $\exists x_i$ depends on,
  such that $B$ is in their scope. Following the translation of the
  $\forall$ rules above, we can therefore translate the $\exists$-rule
  instance by the following derivation
  \begin{equation*}
    \vlder{\set{\sw,\fequ}}{\DDeri_3}{ S_0\cons{S_1\cons{\cdots
          S_{l-1}\cons{S_{l}\cons{\dual q_{i,k_1}\vlan\dual
              q_{i,k_2}\vlan\cdots q_{i,k_l}\vlan B'}}\cdots}}}{
      S_0\cons{\dual q_{i,k_1}\vlan S_1\cons{\dual
          q_{i,k_2}\vlan\cdots S_{l-1}\cons{\dual q_{i,k_l}\vlan
            S_{l}\cons{B'}}\cdots}}}
  \end{equation*}
  where $k_1,\ldots,k_l$ are the indices of the universal binders
  on which that $x_i$ depends, and $B'$ is $B$ in which all predicates
  are replaced by a nullary one (step~3 in the frame construction). The
  derivation $\DDeri_3$ can be constructed in the same way as
  $\DDeri_2$.

  Doing this to all instances of the rules $\forall$ and $\exists$ in
  $\Deri$ yields indeed a propositional derivation $\frameof\Deri$
  with
  \hbox{$\fographof{\frameof\Deri}=\frameof{\fographof\Deri}$}. It has
  been shown by Retor\'e~\cite{retore:pomset:RR} and
  rediscovered by Stra{\ss}burger~\cite{dissvonlutz} that
  $\fographof{\frameof\Deri}=\tuple{\graphof{\frameof
      C},\linkingof{\frameof\Deri}}$ cannot contain an induced
  bimatching. By Lemma~\ref{lem:frame}, $\fographof{\Deri}$ does not
  have an induced bimatching either. Furthermore, it follows from the
  definition of $\dsubstof\Deri$ that it is a dualizer for
  $\fographof{\Deri}$.
\end{proof}

\begin{remark}
  There is an alternative path of proving Theorem~\ref{thm:MLS1->fonet}
  by translating $\Deri$ to an $\FOMLL$-proof $\Pi$, observing that
  this process preserves the linking and the dualizer. However, for
  this, we have to extend the construction from the previous subsection to the $\cut$-rule,
  and then show that linking and dualizer of a sequent proof $\Pi$ are
  invariant under cut elimination. This can be done similarly to
  unification nets in~\cite{hughes:unifn}.
\end{remark}

\subsection{From Fonets to $\FOMLL$ Proofs}

Now we are going to show how from a given fonet
$\tuple{\gC,\linkingof\gC}$ we can construct a sequent proof $\Pi$ in
$\FOMLL$ such that $\fographof\Pi=\tuple{\gC,\linkingof\gC}$. In the
proof net literature, this operation is also called
\emph{sequentialization}. The basic idea behind our sequentialization
is to use the frame of $\gC$, to which  we can apply the \emph{splitting tensor
theorem}, and then reconstruct the sequent proof $\Pi$.

Let $\Gamma$ be a propositional sequent and $\linkingof\Gamma$ be a
linking for $\graphof\Gamma$. A conjunction formula $A\vlan B$ is
\bfit{splitting} or a \bfit{splitting tensor} if
$\Gamma=\Gamma',A\vlan B,\Gamma''$ and
$\mathord{\linkingof\Gamma}=\mathord{\linking_1}\cup\linking_2$, such
that $\linking_1$ is a linking for $\graphof{\Gamma',A}$ and
$\linking_2$ is a linking for $\graphof{B,\Gamma''}$, i.e., removing
the $\vlan$ from $A\vlan B$ splits the linked fograph $\tuple{
  \graphof\Gamma,\linkingof\Gamma}$ into two fographs.
We say that $\tuple{
  \graphof\Gamma,\linkingof\Gamma}$ is \bfit{mixed} iff 
$\Gamma=\Gamma',\Gamma''$ and
$\mathord{\linkingof\Gamma}=\mathord{\linking_1}\cup\linking_2$, such
that $\linking_1$ is a linking for $\graphof{\Gamma'}$ and
$\linking_2$ is a linking for $\graphof{\Gamma''}$.
Finally, $\tuple{
  \graphof\Gamma,\linkingof\Gamma}$ is \bfit{splittable} if it is mixed or has a splitting tensor.

\begin{thm}
  \label{thm:splitting}
  Let $\Gamma$ be a $\fff$-free propositional sequent containing only atoms
  and $\vlan$-formulas, and $\linkingof\Gamma$ be a
  linking for $\graphof\Gamma$. If $\tuple{
    \graphof\Gamma,\linkingof\Gamma}$ does not induce a bimatching then it is splittable.
\end{thm}

This is the well-known splitting-tensor-theorem
\cite{girard:87,danos:regnier:89}, adapted for the presence of
$\mix$. In the setting of linked cographs, it has first been proved by
Retor\'e~\cite{retore:03,retore:99} and then rediscovered by Hughes~\cite{hughes:pws}.
We use it now for our sequentialization:

\begin{thm}
  \label{thm:fonet->MLL1}
  Let $\tuple{\gC,\linkingof\gC}$ be a fonet, and let $\Gamma$ be a sequent with $\graphof\Gamma=\gC$. Then there is an
  $\FOMLL$-proof $\Pi$ of $\Gamma$, such that
  $\fographof\Pi=\tuple{\gC,\linkingof\gC}$.
\end{thm}

\begin{proof}
  Let $\dsubstof\gC$ be the dualizer of $\tuple{\gC,\linkingof\gC}$.  We
  proceed by induction on the size of $\Gamma$ (i.e., the number of
  symbols in it, without counting the commas). If $\Gamma$ contains a
  formula with $\vlor$-root, or a formula $\forall x.A$, we can
  immediately apply the $\vlor$-rule or the $\forall$-rule of $\FOMLL$
  and proceed by induction hypothesis. If $\Gamma$ contains a formula
  $\exists x.A$ such that the corresponding binder $\single x$ in $\gC$
  has no dependency, then we can apply the $\exists$-rule, choosing the
  term $t$ as determined by $\dsubstof\gC$, and proceed by induction
  hypothesis.  Hence, we can now assume that $\Gamma$ contains only
  atoms, $\vlan$-formulas, or formulas of shape $\exists x.A$, where
  the vertex $\single x$ has dependencies. Then the frame
  $\tuple{\graphof{\frameof\Gamma},\linkingof{\frameof\Gamma}}$ does
  not induce a bimatching and contains only atoms and
  $\vlan$-formulas, and is therefore splittable. If it is mixed, then
  we can apply the $\mix$-rule to $\Gamma$ and apply the induction
  hypothesis to the two components. If it is not mixed then there must
  be a splitting tensor. If the splitting $\vlan$ is already in
  $\Gamma$, then we can apply the $\vlan$-rule and proceed by
  induction hypothesis on the two branches. However, if
  $\frameof\Gamma$ is not mixed and all splitting tensors are
  $\vlan$-formulas introduced in step~1 of the frame construction,
  then we get a contradiction as in that case there must be a $\vlor$- or $\forall$-formula in $\Gamma$.
\end{proof}

\subsection{From Fonets to $\FOMLS$ Proofs}

We can now straightforwardly obtain the same result for $\FOMLS$:

\begin{thm}
  \label{thm:fonet->MLS1}
  Let $\tuple{\gC,\linkingof\gC}$ be a fonet, and let $C$ be a formula with $\graphof C=\gC$. Then there is
  a derivation $\vlder{\FOMLS}{\Deri}{C}{\ttt}$ such that $\fographof\Deri=\tuple{\gC,\linkingof\gC}$.
\end{thm}

\begin{proof}
  We apply Theorem~\ref{thm:fonet->MLL1} to obtain a sequent
  proof~$\Pi$ of \hbox{$\sqn C$} with
  $\fographof\Pi=\tuple{\gC,\linkingof\gC}$. Then we apply
  Lemma~\ref{lem:MLL1->MLS1}, observing that the translation from
  $\FOMLL$ to $\FOMLS$ preserves linking and dualizer.
\end{proof}

\begin{remark}
  Note that it is also possible to do a direct ``sequentialization''
  into the deep inference system $\FOMLS$, using the techniques
  presented in \cite{dissvonlutz} and \cite{str:MLL2}.
\end{remark}


\section{Skew Bifibrations and Resource Management}
\label{sec:skew}

In this section we establish the relation between skew bifibrations and derivations in
$\set{\wrD,\wfaD,\acD,\cfaD,\me,\mfaD,\mexD,\fequ}$. However, 
if a derivation $\Deri$ contains instances of the rules $\cfaD$, $\mfaD$, and
$\mexD$ we can no longer naively define the rectification
$\rectif\Deri$ as in the previous section for $\FOMLS$, as these two
rules cannot be applied if premise and conclusion are rectified. For
this reason we define here rectified versions $\rectif\cfaD$,  $\rectif\mfaD$ and
$\rectif\mexD$, shown below:\vadjust{\vskip-1ex}
\begin{equation*}
  \scalebox{1}{$
  \vlinf{\rectif\cfaD}{}{\Scons{\forall x.Ax}}{\Scons{\forall y.\forall x.Ax}}
  \hskip3em
  \begin{array}{c}
    \vlinf{\rectif\mfaD}{}{\Scons{\forall x.(Ax\vlor Bx)}}{\Scons{(\forall y.Ay)\vlor(\forall z.Bz)}}
    \\ \\[-1ex]
    \vlinf{\rectif\mexD}{}{\Scons{\exists x.(Ax\vlor Bx)}}{\Scons{(\exists y.Ay)\vlor(\exists z.Bz)}}
  \end{array}
  $}
\end{equation*}
Here, we use the notation $A\cdot$ for a formula $A$ with occurrences
of a placeholder $\cdot$ for a variable. Then $Ax$ stands for the
results of replacing that placeholder with $x$, and also indicating
that $x$ must not occur in $A\cdot$. Then $\forall x.Ax$ and $\forall
y.Ay$ are the same formula modulo renaming of the bound variable bound
by the outermost $\forall$-quantifier. We also demand that the variables $x$, $y$, and $z$ do
not occur in the context $\Sconhole$.

Note that in an instance of $\rectif\mfaD$ or $\rectif\mexD$ (as shown
above), we can have $x=y$ or $x=z$, but not both if the premise is
rectified. If $x=y$ and $x=z$ we have $\mfaD$ and
$\mexD$ as special cases of $\rectif\mfaD$ and $\rectif\mexD$, respectively. And similarly,  if $x=y$ then $\cfaD$
is a special case of $\rectif\cfaD$.

For a derivation $\Deri$ in
$\set{\wrD,\wfaD,\acD,\cfaD,\me,\mfaD,\mexD,\fequ}$, we can now construct the
\bfit{rectification} $\rectif\Deri$ by rectifying each line of $\Deri$, yielding
a derivation in
$\set{\wrD,\wfaD,\acD,\rectif\cfaD,\me,\rectif\mfaD,\rectif\mexD,\fequ}$.

For each instance $\vlinf{\rr}{}{P}{Q}$ of an inference rule in
$\set{\wrD,\wfaD,\acD,\rectif\cfaD,\me,\rectif\mfaD,\rectif\mexD,\fequ}$ we can
define the \bfit{induced map} $\mapof{\rr}\colon\vgraphof
Q\to\vgraphof P$ which acts as the identity for
$\rr\in\set{\me,\fequ}$ and as the canonical injection for
$\rr\in\set{\wrD,\wfaD}$. For $\rr=\acD$ it maps the vertices
corresponding to the two atoms in the premise to the vertex of the
contracted atom in the conclusion, and for
$\rr\in\set{\rectif\cfaD,\rectif\mfaD,\rectif\mexD}$ it maps the two vertices
corresponding to the quantifiers in the premise to the one in the
conclusion (and acts as the identity on all other vertices).
For a
derivation $\Deri$ in
$\set{\wrD,\wfaD,\acD,\rectif\cfaD,\me,\rectif\mfaD,\rectif\mexD,\fequ}$ we can then define
the \bfit{induced map} $\mapof\Deri$ as the composition of the induced
maps of the rule instances in $\Deri$.

\begin{lemma}\label{lem:cw->rectif}
  Let $\downsmash{\vlderivation{
    \vlde{\set{\wrD,\wfaD,\acD,\cfaD,\me,\mfaD,\mexD,\fequ}}{\Deri}{B}{
      \vlhy{A}}}}$ be given.
  Then there is a rectified derivation
  $\vlderivation{
    \vlde{\set{\wrD,\wfaD,\acD,\rectif\cfaD,\me,\rectif\mfaD,\rectif\mexD,\fequ}}{\rectif\Deri}{\rectif
      B}{ \vlhy{\rectif A}}}$, such that the induced maps
  $\mapof\Deri\colon\graphof{A}\to\graphof{B}$ and
  $\mapof{\rectif\Deri}\colon\graphof{\rectif A}\to\graphof{\rectif
    B}$ are equal up to a variable renaming of the vertex labels.
\end{lemma}

\begin{proof}
  Immediate from the definition.
\end{proof}


\subsection{From Contraction and Weakening to Skew Bifibrations}

\begin{lemma}\label{lem:cw->skew}
  Let $\vlderivation{
    \vlde{\set{\wrD,\wfaD,\acD,\rectif\cfaD,\me,\rectif\mfaD,\rectif\mexD,\fequ}}{\Deri}{B}{
      \vlhy{A}}}$ be a rectified derivation. Then the induced map
  $\mapof\Deri\colon\graphof{A}\to\graphof{B}$ is a skew bifibration.
\end{lemma}

Before we show the proof of this lemma, we introduce another useful
concept: the \bfit{propositional encoding} $\PE{A}$ of a formula $A$,
which is a propositional formula with the property that $\graphof{\PE
  A}=\graphof{A}$. For this, we introduce new propositional variables
that have the same names as the (first-order) variables
$x\in\VAR$. Then $\PE{A}$ is defined inductively by:
\begin{equation*}
  \begin{array}{r@{\;=\;}l}
    \PE{a} & a\\
    \PE{(A \cor B)} &  \PE{A} \cor \PE{B} \\
    \PE{(A \cand B)} & \PE{A} \cand \PE{B}
  \end{array}
  \qquad
  \begin{array}{r@{\;=\;}l}
    \PE{(\forall x A)} & x \cor \PE{A}\\
    \PE{(\exists x A)} & x \cand \PE{A}
  \end{array}
\end{equation*}

\begin{lemma}
  \label{lem:PE}
  For every formula $A$, we have $\graphof{\PE A}=\graphof{A}$.
\end{lemma}

\begin{proof}
  A straightforward induction on $A$. 
\end{proof}

We use $\PE\fequ$ to denote the restriction of $\fequ$ to
propositional formulas, i.e., the first two lines in~\eqref{eq:fequ}.

\begin{proof}[Proof of Lemma~\ref{lem:cw->skew}]
  First, observe that for every inference rule
  $\rr\in\set{\wrD,\wfaD,\acD,\rectif\cfaD,\me,\rectif\mfaD,\rectif\mexD,\fequ}$
  the induced map $\mapof{\rr}\colon\vgraphof Q\to\vgraphof P$ defines
  an existential-preserving graph homomorphism $\graphof Q\to\graphof
  P$ and a fibration on the corresponding binding graphs.  Therefore,
  their composition $\mapof\Deri$ has the same properties of fibration.

  For showing that it is also a skew fibration, we construct for
  $\Deri$ its propositional encoding $\PE\Deri$ by
  translating every line into its propositional encoding.
  The instances
  of the rules $\rectif\mfaD$ and $\rectif\mexD$ are replaced by:
  \begin{equation*}
    \scalebox{.85}{$
    \vlderivation{
      \vlin{\rectif\acD}{}{\Scons{x\vlor(\PE{(Ax)}\vlor\PE{(Bx)})}}{
        \vlin{\fequ}{}{\Scons{(y\vlor z)\vlor(\PE{(Ay)}\vlor\PE{(Bz)})}}{
          \vlhy{\Scons{(y\vlor\PE{(Ay)})\vlor (z\vlor\PE{(Bz)})}}}}}
    \qquad\!
    \vlderivation{
      \vlin{\rectif\acD}{}{\Scons{x\vlan(\PE{(Ax)}\vlor\PE{(Bx)})}}{
        \vlin{\me}{}{\Scons{(y\vlor z)\vlan(\PE{(Ay)}\vlor\PE{(Bz)})}}{
          \vlhy{\Scons{(y\vlan\PE{(Ay)})\vlor (z\vlan\PE{(Bz)})}}}}}
    $}
  \end{equation*}
  respectively, where $\rectif\acD$ is a $\acD$ that renames the
  variables---the propositional variable, as well as the first-order
  variable of the same name---as everything is rectified, there is no
  ambiguity here. Any instance of a rule $\wD$, $\acD$, $\me$, or
  $\fequ$ is translated to an instance of the same rule,
  $\rectif\cfaD$ is translated to $\rectif\acD$, and $\wfaD$ is
  translated to $\wD$.

  This gives us a derivation $\vlderivation{
    \vlde{\set{\wrD,\acD,\rectif\acD,\me,\PE\fequ}}{\PE\Deri}{\PE B}{
      \vlhy{\PE A}}}$ such that $\mapof{\PE\Deri}=\mapof\Deri$. It has
  been shown in~\cite{str:07:RTA} that $\mapof{\PE\Deri}$ is a skew
  fibration.
  Hence,
  $\mapof\Deri$ is a skew fibration.
\end{proof}

\subsection{From Skew Bifibrations to Contraction and Weakening}

\begin{lemma}\label{lem:skew->cw}
  Let $\gA$ and $\gB$ be fographs, let $\phi\colon\gA\to\gB$ be a skew
  bifibration, and let $A$ and $B$ be formulas with $\graphof A=\gA$
  and $\graphof B=\gB$. Then there are derivations
  \begin{equation*}
    \vlderivation{
      \vlde{\set{\wrD,\wfaD,\acD,\rectif\cfaD,\me,\rectif\mfaD,\rectif\mexD,\fequ}}{\rectif\Deri}{B}{
        \vlhy{A}}}
    \quand
    \vlderivation{
      \vlde{\set{\wrD,\wfaD,\acD,\cfaD,\me,\mfaD,\mexD,\fequ}}{\Deri}{B}{
        \vlhy{A\rsubstof\phi}}}
  \end{equation*}
  such that $\mapof{\rectif\Deri}=\phi$ and $\rectif{\Deri}$ is a
  rectification of $\Deri$, and $\rsubstof\phi$ is the substitution
  induced by $\phi$.
\end{lemma}

In the proof of this lemma, we make use of the following concept: Let
$\vlder{\sysS}{\DDeri}{Q}{P}$ be a derivation where $P$ and $Q$ are
propositional formulas (possibly using variable $x\in\VAR$ at the
places of atoms).  We say that $\DDeri$ can be \bfit{lifted} to
$\sysS'$ if there are (first-order) formulas $C$ and $D$ such that $P=\PE C$ and
$Q=\PE D$ and there is a derivation $\vlder{\sysS'}{\DDeri'}{D}{C}$.

We say a fograph homomorphism $\phi\colon\gG\to\gH$ is \bfit{full} if
for all $v,w\in\vG$, we have that $\phi(v)\phi(w)\in\eH$ implies
$vw\in\eG$.

\begin{lemma}\label{lem:fullinj->w}
  Let $\phi\colon\gG\to\gH$ be full and injective skew bifibration
  such that $\rsubstof\phi$ is the identity substitution, and let $G$
  and $H$ be formulas with $\graphof G=\gG$ and $\graphof H=\gH$. Then
  there is a derivation $\vlder{\set{\wrD,\wfaD,\fequ}}{\Deri}{H}{G}$.
\end{lemma}

\begin{proof}
   By \cite[Proposition~7.6.1]{str:07:RTA}, we have a derivation
   $\vlder{\set{\wrD,\PE\fequ}}{\DDeri}{\PE H}{\PE G}$. In order to
   lift $\DDeri$, we need to reorganize the instances of
   $\wD$. If $H$ contains a subformula $\forall x.A$ which is not
   present in $G$, the $\wD$-instances in $\DDeri$ could introduce the
   parts of the propositional encoding $x\vlor A$ independently.  We
   say that an instance $\rr_1$ of $\wD$ in $\Deri$ is \emph{in the
   scope} of an instance $\rr_2$ of $\wD$ if $\rr_1$ introduced
   formulas that contain a free variable $x$ (i.e., $x$ occurs in a
   term in a predicate) and $\rr_2$ introduces the atom $x$ as a
   subformula (i.e. the propositional encoding of the binder $x$).  We
   can now permute the $\wD$-instances in $\DDeri$ such that whenever
   a rule instance $\rr_1$ is in the scope of an instance $\rr_2$, then
   $\rr_2$ occurs below $\rr_1$ in $\DDeri$.
   Then we can lift $\DDeri$ stepwise. First, observe that each line of
   $\DDeri$ is $\PE\fequ$-equivalent to the propositional encoding $\PE P$ of a
   first-order formula $P$.  We now have to show
   that each instance of $\wrD$ in $\DDeri$ is indeed the image of a correct
   application of $\wD$ or $\wfaD$ in first-order logic.
   If we have a $\wD$ of the form
   \begin{equation*}
     \scalebox{1}{$
     \vlinf{\wD}{}{\PE S\cons{x\vlor \PE A}}{\PE S\cons{\PE A}}
     \qquor
     \vlinf{\wD}{}{\PE S\cons{(x\vlor \PE B)\vlor \PE A}}{\PE S\cons{\PE A}}
     $}
   \end{equation*}
   then $x$ cannot occur freely in $A$, as otherwise the fibration
   property would be violated. We can therefore lift these instances to
   \begin{equation*}
      \scalebox{1}{$
    \vlinf{\wfaD}{}{S\cons{\forall x. A}}{S\cons{A}}
     \qquor
     \vlinf{\wD}{}{S\cons{(\forall x. B)\vlor A}}{S\cons{A}}
     $}
   \end{equation*}
   respectively.     
  If a  weakening happens inside a subformula $x\vlor \PE C$ or $x\vlan
  \PE C$ in $\DDeri$, then there are the following cases:
  \begin{equation*}
    \scalebox{.9}{$
    \vlinf{\wD}{}{\PE S\cons{x\vlor\PE  D\vlor \PE C}}{\PE S\cons{x\vlor\PE  C}}
    \quad
    \vlinf{\wD}{}{\PE S\cons{x\vlan (\PE D\vlor \PE C)}}{\PE S\cons{x\vlan\PE  C}}
    \quad
    \vlinf{\wD}{}{\PE S\cons{(x\vlor\PE  D)\vlan \PE C}}{\PE S\cons{x\vlan \PE C}}
    $}
  \end{equation*}
  The first two cases can be lifted to
  \begin{equation*}
    \vlinf{\wD}{}{S\cons{\forall x.( D\vlor C)}}{S\cons{\forall x.  C}}
    \quand
     \vlinf{\wD}{}{S\cons{\exists x.( D\vlor C)}}{S\cons{\exists x.  C}}
  \end{equation*}
  respectively. But in the third case,  an $\exists$-quantifier would be
  transformed into an $\forall$-quantifier. But as $\phi$ has to
  preserve existentials, this third case cannot occur. All other situations can be lifted trivially, giving us
  $\vlder{\set{\wrD,\wfaD,\fequ}}{\Deri}{H}{G}$ as desired.
\end{proof}

\begin{lemma}\label{lem:surj->cm}
  Let $\phi\colon\gG\to\gH$ be a surjective skew
  bifibration, and let $G$ and $H$ be formulas with $\graphof G=\gG$
  and $\graphof H=\gH$. Then there is a derivation
  \begin{equation*}
    \vlderivation{
      \vlde{\set{\acD,\cfaD,\me,\mfaD,\mexD,\fequ}}{\Deri}{H}{
        \vlhy{G\rsubstof\phi}}}
  \end{equation*}
  where  $\rsubstof\phi$ is the substitution
  induced by $\phi$.
\end{lemma}

\begin{proof}
  By \cite[Proposition~7.5]{str:ral:tableaux19}, there is a derivation
  $\vlder{\set{\acD,\me,\PE\fequ}}{\DDeri}{\PE
    H}{\PE{(G\substof\phi)}}$.  We can lift $\DDeri$ to a first-order
  derivation in $\set{\acD,\cfaD,\me,\mfaD,\mexD,\fequ}$, in a similar
  way as in the previous lemma.
  The technical details are in
  Appendix~\ref{app:surj->cm}.
\end{proof}
  
\begin{proof}[Proof of Lemma~\ref{lem:skew->cw}]
  Let $\vB'\subseteq\vB$ be the image of $\phi$, and let $\gB_1$ be
  the subgraph of $\gB$ induced by $\vB'$. Hence, we have two maps
  $\phi''\colon\gA\to\gB_1$ being a surjection and $\phi'\colon
  \gB_1\to\gB$ being a full injection.  Both, $\phi'$ and $\phi''$
  remain skew bifibrations.  Furthermore, $\gB_1$ is also a fograph.
  Let $B_1$ be a formula with $\graphof{B_1}=\gB_1$.
    We can apply Lemmas~\ref{lem:fullinj->w}
  and~\ref{lem:surj->cm} to obtain derivations
  \begin{equation*}
    \vlder{\set{\wrD,\wfaD,\fequ}}{\Deri'}{B}{B_1}
    \qquand
     \vlder{\set{\acD,\cfaD,\me,\mfaD,\mexD,\fequ}}{\Deri''}{B_1}{
        A\rsubstof{\phi''}}
  \end{equation*}
  As $\rsubstof{\phi'}$ is the identity, we have
  $\rsubstof{\phi''}=\rsubstof{\phi}$. Hence, the composition of
  $\Deri''$ and $\Deri'$ is the desired derivation $\Deri$.  Then
  $\rectif\Deri$ can be constructed by rectifying $\Deri$, where the
  variables to be used in $A$ are already given. That
  $\phi=\mapof{\rectif\Deri}$ follows immediately from the
  construction.
\end{proof}


\section{Summary and Proof of Main Result}
\label{sec:summary}

The only theorem of Section~\ref{sec:main} that has not yet been
proved is Theorem~\ref{thm:CP-DI} establishing the full correspondence
between decomposed proofs in $\FOKS$ and combinatorial proofs. We show
the proof here, by summarizing the results of the previous two
Sections~\ref{sec:linear} and~\ref{sec:skew}.

\begin{proof}[Proof of Theorem~\ref{thm:CP-DI}]
  First, assume we have a combinatorial proof $\phi\colon\gC\to\gA$
   and a formula $A$ with $\gA=\graphof A$.
  Let $C$ be a formula with $\graphof C=\gC$, and let $\rsubstof\phi$
  be the substitution induced by $\phi$.
  By Lemma~\ref{lem:skew->cw} there is a derivation
  \begin{equation*}
    \vlder{\set{\wrD,\wfaD,\acD,\cfaD,\me,\mfaD,\mexD,\fequ}}{\Deri_2}{A}{C\rsubstof\phi}
  \end{equation*}
  Since $\gC$ is a fonet, we have by Theorem~\ref{thm:fonet->MLS1} a derivation
  \begin{equation*}
    \vlder{\FOMLS}{\Deri'_1}{C}{\ttt}
  \end{equation*}
  This derivation remains valid if we apply the substitution
  $\rsubstof\phi$ to every line in $\Deri'_1$, yielding the derivation
  $\Deri_1$ of $C\rsubstof\phi$ as desired.

  Conversely, assume we have a decomposed derivation
  \begin{equation}
    \label{eq:decomx}
    \vlderivation{
      \vlde{\set{\wrD,\wfaD,\acD,\cfaD,\me,\mfaD,\mexD,\fequ}}{\Deri_2}{A}{
        \vlde{\FOMLS}{\Deri_1}{A'}{
          \vlhy{\ttt}}}}
  \end{equation}
  Then we can transform $\Deri_1$ into a rectified form
  $\rectif\Deri_1$, proving $\rectif A'$. By
  Theorem~\ref{thm:MLS1->fonet}, the linked fograph
  $\fographof{\rectif\Deri_1}=\tuple{\fographof{\rectif
      A'},\linkingof{\rectif\Deri_1}}$ is a fonet.  Then, by
  Lemma~\ref{lem:cw->rectif}, there is a rectified derivation
  $\vlderivation{
    \vlde{\set{\wrD,\wfaD,\acD,\rectif\cfaD,\me,\rectif\mfaD,\rectif\mexD,\fequ}}{\rectif{\Deri_2}}{\rectif
      A}{ \vlhy{\rectif{A'}}}}$ whose induced map
  $\mapof{\rectif{\Deri_2}}\colon\graphof{\rectif{A'}}\to\graphof{\rectif
    A}$ is the same as the induced map
  $\mapof{\Deri_2}\colon\graphof{A'}\to\graphof{A}$ of $\Deri_2$. By
  Lemma~\ref{lem:cw->skew}, this map is a skew bifibration. Hence, we
  have a combinatorial proof $\phi\colon\gC\to\graphof{A}$ with
  $\gC=\graphof{\rectif{A'}}$.
\end{proof}

Note that Theorem~\ref{thm:CP-DI} shows at the same time soundness, completeness, and full completeness, as
\begin{enumerate}
\item every proof in $\FOKS$ can be translated into a combinatorial proof, and
\item every combinatorial proof is the image of a $\FOKS$-proof under that translation.
\end{enumerate}


\section{Conclusion}

We uncovered a close correspondence between
first-order combinatorial proofs and decomposed deep inference
derivations of system $\FOKS$, and showed that every proof in
$\FOKS$ has such a decomposed form.

The most surprising discovery for us was that all technical
difficulties in our work could be reduced (in a non-trivial way) to
the propositional setting.

The obvious next step in our research is to investigate proof
composition and normalisation of first-order combinatorial
proofs. Even in the propositional setting, the normalisation of
combinatorial proofs is underdeveloped. There exist two different
procedures for cut elimination for combinatorial proofs in classical
propositional logic~\cite{hughes:invar,str:fscd17}, but both have
their insufficiencies, and have not been extended to other
logics.

We hope to garner new insights on the normalisation of
classical first-order proofs through our work on combinatorial proofs.






%

\bibliographystyle{IEEEtran}
\bibliography{refs}







\newpage
\appendix

\subsection{Proof of Theorem~\ref{thm:LK1-decompose}}
\label{app:LK1-decompose}

\begin{proof}[Proof of Theorem~\ref{thm:LK1-decompose}]

Write $\fv(A)$ for the set of variables which occur free in $A$.

Note that the instances of $\wrD,\cD$ in $\Deri_2$ are deep, but inside sequent contexts.

First, if an instance of $\vlinf{\weakr}{}{\sqn{\Gamma,A}}{\sqn{\Gamma}}$
 is followed by a rule in which $A$ is not in the principal
formula, it can be permuted downwards.
Otherwise, the proof can be transformed using the following rewriting rules.

\begin{equation*}
\vlderivation{
  \vliin{\vlan}{}{\sqns{\Gamma, A \vlan B, \Delta}}{
    \vlin{\weakr}{}{\sqns{\Gamma, A}}{
      \vlhy{\sqns{\Gamma}}}}{
    \vlhy{\sqns{B, \Delta}}}}
\leadsto
\vlderivation{
  \vlin{\weakr}{}{\sqns{\Gamma, A \vlan B, \Delta}}{
    \vlin{\weakr}{}{\vlvdots}{
      \vlhy{\sqns{\Gamma}}}}}
\end{equation*}
\begin{equation*}
\vlderivation{
  \vlin{\vlor}{}{\sqns{\Gamma, A \vlor B}}{
    \vlin{\weakr}{}{\sqns{\Gamma, A, B}}{
      \vlhy{\sqns{\Gamma, A}}}}}
\leadsto
\vlderivation{
  \vlin{\wrD}{}{\sqns{\Gamma, A \vlor B}}{
    \vlhy{\sqns{\Gamma, A}}}}
\end{equation*}
\begin{equation*}
\vlderivation{
  \vlin{\exists}{}{\sqns{\Gamma, \exists x.A}}{
    \vlin{\weakr}{}{\sqns{\Gamma, A\ssubst{x}{t}}}{
      \vlhy{\sqns{\Gamma}}}}}
\leadsto
\vlderivation{
  \vlin{\weakr}{}{\sqns{\Gamma, \exists x.A}}{
    \vlhy{\sqns{\Gamma}}}}
\end{equation*}
\begin{equation*}
\vlderivation{
  \vlin{\forall}{}{\sqns{\Gamma, \forall x.A}}{
    \vlin{\weakr}{}{\sqns{\Gamma, A}}{
      \vlhy{\sqns{\Gamma}}}}}
\leadsto
\vlderivation{
  \vlin{\weakr}{}{\sqns{\Gamma, \forall x.A}}{
    \vlhy{\sqns{\Gamma}}}}
\end{equation*}
\begin{equation*}
\vlderivation{
  \vlin{\conr}{}{\sqns{\Gamma, A}}{
    \vlin{\weakr}{}{\sqns{\Gamma, A, A}}{
      \vlhy{\sqns{\Gamma, A}}}}}
\leadsto
\vlderivation{
  \vlhy{\sqns{\Gamma, A}}}
\end{equation*}
Note that in the case of $\vlor$, we use the deep rule $\wrD$ which can be
permuted under all the rules. By using these rewriting rules, we can eventually get
a derivation with all the instances of $\weakr$ and $\wrD$ at the bottom. Now observe
that the instances of $\conr$ in $\Deri$ can be transformed using the following
rule:

\begin{equation*}
\vlderivation{
  \vlin{\conr}{}{\sqns{\Gamma, A}}{
    \vlhy{\sqns{\Gamma, A, A}}}}
\leadsto
\vlderivation{
  \vlin{\cD}{}{\sqns{\Gamma, A}}{
    \vlin{\vlor}{}{\sqns{\Gamma, A \vlor A}}{
      \vlhy{\sqns{\Gamma, A, A}}}}}
\end{equation*}

Knowing that $\cD$ can be permuted under all the rules of $\FOMLL$, we
eventually obtain a derivation:
\begin{equation*}
\vlderivation{
  \vlde{\set{\weakr, \wrD,\cD,\fequ}}{\Deri'_2}{\sqns{\Gamma}}{
    \vltrl{}{\FOMLL}{\Deri'_1}{\sqns{\Gamma'}}{
      \vlhy{\quad}}{
      \vlhy{}}{
      \vlhy{}}}}
\end{equation*}
Note that $\fequ$ is required here since the permutation of formulas is implicit
in $\FOMLL$.

By transforming each sequent of $\Deri'_2$ into its corresponding formula, and by
considering the following rewriting rule:
\begin{equation*}
\vlderivation{
  \vlin{\weakr}{}{\sqns{\Gamma, A}}{
    \vlhy{\sqns{\Gamma}}}}
\leadsto
\vlderivation{
  \vlin{\wrD}{}{\sqns{\form{\Gamma} \vlor A}}{
    \vlhy{\sqns{\form{\Gamma}}}}}
\end{equation*},
we obtain a derivation
\begin{equation*}
\vlderivation{
  \vlde{\set{\wrD,\cD,\fequ}}{\Deri_2}{\sqns{\form{\Gamma}}}{
    \vltrl{}{\FOMLL}{\Deri_1}{\sqns{\form{\Gamma'}}}{
      \vlhy{\quad}}{
      \vlhy{}}{
      \vlhy{}}}}
\end{equation*}
where $\Deri_1$ can be obtained from $\Deri'_1$
by applying the $\vlor$ rule.
\end{proof}

\subsection{Rule permutation for the proof of Lemma~\ref{lem:cw-decomposition}}
\label{app:cw-decomposition}
We construct a rewriting system based on rule permutation on derivations in $\set{\wrD, \wfaD, \acD,
\cfaD, \me, \mfaD, \mexD, \fequ}$ that allows us to reach a derivation of the
form 
\begin{equation*}
    \vlderivation{
      \vlde{\set{\wrD,\wfaD,\fequ}}{}{B}{
        \vlde{\set{\acD,\cfaD}}{}{B'}{
          \vlde{\set{\me,\mfaD,\mexD,\fequ}}{}{A'}{
            \vlhy{A}}}}}
\end{equation*}
from any derivation. Intuitively, we want to move all the instances of $\rr \in
\set{\wrD, \wfaD}$ downwards and all the instances of $\rr' \in \set{\me,
\mfaD, \mexD}$ upwards.

We first study the interactions between two rules. Certain cases are unsolved at
this stage, and they are considered later when we study the interactions between
two non-$\fequ$ rule instances separated by $\fequ$. Only non-trivial cases are
presented here:
\begin{itemize}
\item $\rr_1/\rr_2$, where $\rr_1 \in \set{\wrD, \wfaD}$ and $\rr_2 \in
\set{\acD, \cfaD, \me, \mfaD, \mexD}$:

\begin{equation*}
\vlderivation{
  \vlin{\acD}{}{a}{
    \vlin{\wrD}{}{a \vlor a}{ 
      \vlhy{a}}}}
\leadsto
\vlderivation{
  \vlhy{a}}
\end{equation*}

\begin{equation*}
\vlderivation{
  \vlin{\me}{}{(A \vlor B) \vlan (C \vlor D)}{
    \vlin{\wrD}{}{(A \vlan C) \vlor (B \vlan D)}{
      \vlhy{A \vlan C}}}}
\leadsto 
\vlderivation{
  \vlin{\wrD}{}{(A \vlor B) \vlan (C \vlor D)}{
    \vlin{\wrD}{}{(A \vlor B) \vlan C}{
      \vlhy{A \vlan C}}}}
\end{equation*}

\begin{equation*}
\vlderivation{
  \vlin{\mfaD}{}{\forall x.(A \vlor B)}{
    \vlin{\wrD}{}{(\forall x.A) \vlor (\forall x.B)}{
      \vlhy{\forall x.A}}}}
\leadsto
\vlderivation{
  \vlin{\wrD}{}{\forall x.(A \vlor B)}{
    \vlhy{\forall x.A}}}
\end{equation*}

\begin{equation*}
\vlderivation{
  \vlin{\cfaD}{}{\forall x.A}{
    \vlin{\wfaD}{}{\forall x.\forall x.A}{
      \vlhy{\forall x.A}}}}
\leadsto
\vlderivation{
  \vlhy{\forall x.A}}
\end{equation*}

\begin{equation*}
\vlderivation{
  \vlin{\mfaD}{}{\forall x.(A \vlor B)}{
    \vlin{\wfaD}{}{(\forall x.A) \vlor (\forall x.B)}{
      \vlhy{A \vlor (\forall x.B)}}}}
\leadsto
\vlderivation{
  \vlin{\fequ}{}{\forall x.(A \vlor B)}{
    \vlhy{A \vlor (\forall x.B)}}} 
\end{equation*}
where in the last case, $x$ is not free in $A$.

\item $\rr_1/\rr_2$, where $\rr_1 \in \set{\acD, \cfaD}$ and $\rr_2 \in
\set{\me, \mfaD, \mexD}$: 
  \begin{equation*}
  \scalebox{.9}{$
  \vlderivation{
    \vlin{\mfaD}{}{\Scons{\forall x.(A \vlor B)}}{
      \vlin{\cfaD}{}{\Scons{(\forall x.A) \vlor (\forall x.B)}}{
        \vlhy{\Scons{(\forall x.\forall x.A) \vlor (\forall x.B)}}}}}
  \leadsto
  \vlderivation{
    \vlin{\mfaD}{}{\Scons{\forall x.(A \vlor B)}}{
      \vlin{\fequ}{}{\Scons{(\forall x.A) \vlor (\forall x.B)}}{
        \vlin{\mfaD}{}{\Scons{\forall x.(\forall x.A \vlor B)}}{
          \vlhy{\Scons{(\forall x.\forall x.A) \vlor (\forall x.B)}}}}}}$}
  \end{equation*}

\item $\cfaD/\fequ$:

\begin{equation*}
\vlderivation{
  \vlin{\fequ}{}{\forall y.\forall x.A}{
    \vlin{\cfaD}{}{\forall x.\forall y.A}{
      \vlhy{\forall x.\forall x.\forall y.A}}}}
\leadsto
\vlderivation{
  \vlin{\cfaD}{}{\forall y.\forall x.A}{
    \vlin{\fequ}{}{\forall y.\forall x.\forall x.A}{
      \vlhy{\forall x.\forall x.\forall y.A}}}}
\end{equation*}

\begin{equation*}
\vlderivation{
  \vlin{\fequ}{}{(\forall x.A) \vlor B}{
    \vlin{\cfaD}{}{\forall x.(A \vlor B)}{
      \vlhy{\forall x.\forall x.(A \vlor B)}}}}
\leadsto
\vlderivation{
  \vlin{\cfaD}{}{(\forall x.A) \vlor B}{
    \vlin{\fequ}{}{(\forall x.\forall x. A) \vlor B}{
      \vlhy{\forall x.\forall x.(A \vlor B)}}}}
\end{equation*}

\begin{equation*}
\vlderivation{
  \vlin{\fequ}{}{\forall x.(A \vlor B)}{
    \vlin{\cfaD}{}{(\forall x.A) \vlor B}{
      \vlhy{(\forall x.\forall x.A) \vlor B}}}}
\leadsto
\vlderivation{
  \vlin{\cfaD}{}{\forall x.(A \vlor B)}{
    \vlin{\fequ}{}{\forall x.\forall x.(A \vlor B)}{
      \vlhy{(\forall x.\forall x.A) \vlor B}}}}
\end{equation*}
where in the last two cases, $x$ is not free in $B$.

\item $\wrD/\fequ$:

\begin{equation*}
\vlderivation{
  \vlin{\fequ}{}{B \vlor A}{
    \vlin{\wrD}{}{A \vlor B}{
      \vlhy{A}}}}
\end{equation*}

\begin{equation*}
\vlderivation{
  \vlin{\fequ}{}{A \vlor (B \vlor C)}{
    \vlin{\wrD}{}{(A \vlor B) \vlor C)}{
      \vlhy{A \vlor C}}}}
\end{equation*}

\begin{equation*}
\vlderivation{
  \vlin{\fequ}{}{(\forall x.A) \vlor B}{
    \vlin{\wrD}{}{\forall x.(A \vlor B)}{
      \vlhy{\forall x.A}}}}
\leadsto
\vlderivation{
  \vlin{\wrD}{}{(\forall x. A) \vlor B}{
    \vlhy{\forall x.A}}}
\end{equation*}

\begin{equation*}
\vlderivation{
  \vlin{\fequ}{}{(\forall x.A) \vlor B}{
    \vlin{\wrD}{}{\forall x.(B \vlor A)}{
      \vlhy{\forall x.B}}}}
\end{equation*}

\begin{equation*}
\vlderivation{
  \vlin{\fequ}{}{\forall x.(A \vlor B)}{
    \vlin{\wrD}{}{(\forall x.A) \vlor B}{
      \vlhy{\forall x.A}}}}
\leadsto
\vlderivation{
  \vlin{\wrD}{}{\forall x.(A \vlor B)}{
    \vlhy{\forall x.A}}}
\end{equation*}

\begin{equation*}
\vlderivation{
  \vlin{\fequ}{}{\forall x.(A \vlor B)}{
    \vlin{\wrD}{}{B \vlor (\forall x.A)}{
      \vlhy{B}}}}
\end{equation*}
where in the last four cases, $x$ is not free in $B$.
\item $\wfaD/\fequ$:

In the following two cases, we assume $x \neq y$ (otherwise they are trivial).

\begin{equation*}
\scalebox{.9}{$
\vlderivation{
  \vlin{\fequ}{}{\forall y.\forall x.A}{
    \vlin{\wfaD}{(x \notin \fv(\forall y.A))}{\forall x.\forall y.A}{
      \vlhy{\forall y.A}}}}
\leadsto
\vlderivation{
  \vlin{\wfaD}{(x \notin \fv(A))}{\forall y.\forall x.A}{
    \vlhy{\forall y.A}}}$}
\end{equation*}

\begin{equation*}
\scalebox{.9}{$
\vlderivation{
  \vlin{\fequ}{}{\forall x.\forall y.A}{
    \vlin{\wfaD}{(x \notin \fv(A))}{\forall y.\forall x.A}{
      \vlhy{\forall y.A}}}}
\leadsto
\vlderivation{
  \vlin{\wfaD}{(x \notin \fv(\forall y.A))}{\forall x.\forall y.A}{
    \vlhy{\forall y.A}}}$}
\end{equation*}

\begin{equation*}
\vlderivation{
  \vlin{\fequ}{}{(\forall x.A) \vlor B}{
    \vlin{\wfaD}{}{\forall x.(A \vlor B)}{
      \vlhy{A \vlor B}}}}
\leadsto
\vlderivation{
  \vlin{\wfaD}{}{(\forall x.A) \vlor B}{
    \vlhy{A \vlor B}}}
\end{equation*}

\begin{equation*}
\vlderivation{
  \vlin{\fequ}{}{\forall x.(A \vlor B)}{
    \vlin{\wfaD}{}{(\forall x.A) \vlor B}{
      \vlhy{A \vlor B}}}}
\leadsto
\vlderivation{
  \vlin{\wfaD}{}{\forall x.(A \vlor B)}{
    \vlhy{A \vlor B}}}
\end{equation*}
where in the last two cases, the constraint on $x$ on the left-hand side implies
that of the right-hand side.

\item $\fequ/\cfaD$:

\begin{equation*}
\vlderivation{
  \vlin{\cfaD}{}{\forall x.\forall y.A}{
    \vlin{\fequ}{}{\forall x.\forall x.\forall y. A}{
      \vlhy{\forall x.\forall y.\forall x.A}}}}
\end{equation*}

\begin{equation*}
\vlderivation{
  \vlin{\cfaD}{}{\forall y.\forall x.A}{
    \vlin{\fequ}{}{\forall y.\forall x.\forall x.A}{
      \vlhy{\forall x.\forall y.\forall x.A}}}}
\end{equation*}

\begin{equation*}
\vlderivation{
  \vlin{\cfaD}{}{(\forall x.A) \vlor B}{
    \vlin{\fequ}{(x \notin \fv(B))}{(\forall x.\forall x.A) \vlor B}{
      \vlhy{\forall x.((\forall x.A) \vlor B)}}}}
\end{equation*}

\begin{equation*}
\vlderivation{
  \vlin{\cfaD}{}{\forall x.(A \vlor B)}{
    \vlin{\fequ}{(x \notin \fv(B))}{\forall x.\forall x.(A \vlor B)}{
      \vlhy{\forall x.((\forall x.A) \vlor B)}}}}
\end{equation*}

\item $\fequ/\me$:

\begin{equation*}
\vlderivation{
  \vlin{\me}{}{(A \vlor B) \vlan (C \vlor D)}{
    \vlin{\fequ}{}{(A \vlan C) \vlor (B \vlan D)}{
      \vlhy{(C \vlan A) \vlor (B \vlan D)}}}}
\end{equation*}

\begin{equation*}
\vlderivation{
  \vlin{\me}{}{(A \vlor B) \vlan (C \vlor D)}{
    \vlin{\fequ}{}{(A \vlan C) \vlor (B \vlan D)}{
      \vlhy{(B \vlan D) \vlor (A \vlan C)}}}}
\leadsto
\vlderivation{
  \vlin{\fequ}{}{(A \vlor B) \vlan (C \vlor D)}{
    \vlin{\me}{}{(B \vlor A) \vlan (D \vlor C)}{
      \vlhy{(B \vlan D) \vlor (A \vlan C)}}}}
\end{equation*}

\begin{equation*}
\vlderivation{
  \vlin{\me}{}{(A \vlor B) \vlan ((C \vlan E) \vlor D)}{
    \vlin{\fequ}{}{(A \vlan (C \vlan E)) \vlor (B \vlan D)}{
      \vlhy{((A \vlan C) \vlan E) \vlor (B \vlan D)}}}}
\end{equation*}

\begin{equation*}
\vlderivation{
  \vlin{\me}{}{\forall x.((A \vlor B) \vlan (C \vlor D))}{
    \vlin{\fequ}{(x \notin \fv(B \vlan D))}{\forall x.((A \vlan C) \vlor (B \vlan D))}{
      \vlhy{(\forall x.(A \vlan C)) \vlor (B \vlan D)}}}}
\end{equation*}

\item $\fequ/\mfaD$:

\begin{equation*}
\vlderivation{
  \vlin{\mfaD}{}{\forall x.(A \vlor B)}{
    \vlin{\fequ}{}{(\forall x.A) \vlor (\forall x.B)}{
      \vlhy{(\forall x.B) \vlor (\forall x.A)}}}}
\leadsto
\vlderivation{
  \vlin{\fequ}{}{\forall x.(A \vlor B)}{
    \vlin{\mfaD}{}{\forall x.(B \vlor A)}{
      \vlhy{(\forall x.B) \vlor (\forall x.A)}}}}
\end{equation*}

\begin{equation*}
\vlderivation{
  \vlin{\mfaD}{}{\forall x.((\forall y.A) \vlor B)}{
    \vlin{\fequ}{}{(\forall x.\forall y.A) \vlor (\forall x.B)}{
      \vlhy{(\forall y.\forall x.A) \vlor (\forall x.B)}}}}
\end{equation*}

\begin{equation*}
\vlderivation{
  \vlin{\mfaD}{}{\forall x.(A \vlor B)}{
    \vlin{\fequ}{}{(\forall x.A) \vlor (\forall x.B)}{
      \vlhy{\forall x.(A \vlor (\forall x.B))}}}}
\end{equation*}

\item $\fequ/\mexD$: similar to $\fequ/\mfaD$
 
\end{itemize}

Interactions between two non-$\fequ$ rules with the presence of $\fequ$ in
between:

\begin{itemize}
\item $\cfaD/\fequ/\rr$ where $\rr \in \set{\me, \mfaD, \mexD}$: First permute
$\cfaD$ under $\fequ$ and then permute $\cfaD$ under $\rr$.

\item $\acD/\fequ/\rr$ where $\rr \in \set{\me, \mfaD, \mexD}$: First permute
$\acD$ under $\fequ$ and then permute $\acD$ under $\rr$.

\item $\wrD/\fequ/\cfaD$:
  \begin{equation*}
  \vlderivation{
    \vlin{\cfaD}{}{(\forall x.A) \vlor B}{
      \vlin{\fequ}{}{(\forall x.\forall x.A) \vlor B}{
        \vlin{\wrD}{}{\forall x.((\forall x.A) \vlor B)}{
          \vlhy{\forall x.\forall x.A}}}}}
  \leadsto
  \vlderivation{
    \vlin{\wrD}{}{(\forall x.A) \vlor B}{
      \vlin{\cfaD}{}{(\forall x.A) \vlor B}{
        \vlhy{\forall x.\forall x.A}}}}
  \end{equation*}
  \begin{equation*}
  \vlderivation{
    \vlin{\cfaD}{}{(\forall x.A) \vlor B}{
      \vlin{\fequ}{}{(\forall x.\forall x.A) \vlor B}{
        \vlin{\wrD}{}{\forall x.(B \vlor (\forall x.A))}{
          \vlhy{\forall x.B}}}}}
  \leadsto
  \vlderivation{
    \vlin{\fequ}{}{(\forall x.A) \vlor B}{
      \vlin{\wrD}{}{\forall x.(B \vlor A)}{
        \vlhy{\forall x.B}}}}
  \end{equation*}

  \begin{equation*}
  \vlderivation{
    \vlin{\cfaD}{}{\forall x.(A \vlor B)}{
      \vlin{\fequ}{}{\forall x.\forall x.(A \vlor B)}{
        \vlin{\wrD}{}{\forall x.((\forall x.A) \vlor B)}{
          \vlhy{\forall x.\forall x.A}}}}}
  \leadsto
  \vlderivation{
    \vlin{\wrD}{}{\forall x.(A \vlor B)}{
      \vlin{\cfaD}{}{\forall x.A}{
        \vlhy{\forall x.\forall x.A}}}}
  \end{equation*}

  \begin{equation*}
  \vlderivation{
    \vlin{\cfaD}{}{\forall x.(A \vlor B)}{
      \vlin{\fequ}{}{\forall x.\forall x.(A \vlor B)}{
        \vlin{\wrD}{}{\forall x.(B \vlor (\forall x.A))}{
          \vlhy{\forall x.B}}}}}
  \leadsto
  \vlderivation{
    \vlin{\fequ}{}{\forall x.(A \vlor B)}{
      \vlin{\wrD}{}{\forall x.(B \vlor A)}{
        \vlhy{\forall x.B}}}}
  \end{equation*}
  where in all four cases, $x$ is not free in $B$.
\item $\wrD/\fequ/\acD$:
  \begin{equation*}
  \vlderivation{
    \vlin{\acD}{}{a \vlor B}{
      \vlin{\fequ}{}{(a \vlor a) \vlor B}{
        \vlin{\wrD}{}{(a \vlor B) \vlor a}{
          \vlhy{a \vlor B}}}}}
  \leadsto
  \vlderivation{
    \vlhy{a \vlor B}}
  \end{equation*}

  \begin{equation*}
  \vlderivation{
    \vlin{\acD}{}{a \vlor B}{
      \vlin{\fequ}{}{(a \vlor a) \vlor B}{
        \vlin{\wrD}{}{a \vlor (a \vlor B)}{
          \vlhy{a}}}}}
  \leadsto
  \vlderivation{
    \vlin{\wrD}{}{a \vlor B}{
      \vlhy{a}}}
  \end{equation*}
  
  \begin{equation*}
  \vlderivation{
    \vlin{\acD}{}{\forall x.a}{
      \vlin{\fequ}{(x \notin \fv(a))}{\forall x.(a \vlor a)}{
        \vlin{\wrD}{}{(\forall x.a) \vlor a}{
          \vlhy{\forall x.a}}}}}
  \leadsto
  \vlderivation{
    \vlhy{\forall x.a}}
  \end{equation*}

  \begin{equation*}
  \vlderivation{
    \vlin{\acD}{}{\forall x.a}{
      \vlin{\fequ}{(x \notin \fv(a))}{\forall x.(a \vlor a)}{
        \vlin{\wrD}{}{a \vlor (\forall x.a)}{
          \vlhy{a}}}}}
  \leadsto
  \vlderivation{
    \vlin{\wfaD}{(x \notin \fv(a))}{\forall x.a}{
      \vlhy{a}}}
  \end{equation*}

\item $\wrD/\fequ/\me$:
  \begin{equation*}
  \vlderivation{
    \vlin{\me}{}{(A \vlor B) \vlan (C \vlor D)}{
      \vlin{\fequ}{}{(A \vlan C) \vlor (B \vlan D)}{
        \vlin{\wrD}{}{(C \vlan A) \vlor (B \vlan D)}{
          \vlhy{C \vlan A}}}}}
  \leadsto
  \vlderivation{
    \vlin{\wrD}{}{(A \vlor B) \vlan (C \vlor D)}{
      \vlin{\wrD}{}{(A \vlor B) \vlan C}{
        \vlin{\fequ}{}{A \vlan C}{
          \vlhy{C \vlan A}}}}}
  \end{equation*}

  \begin{equation*}
  \scalebox{.9}{$
  \vlderivation{
    \vlin{\me}{}{\forall x.((A \vlor B) \vlan (C \vlor D))}{
      \vlin{\fequ}{}{\forall x.((A \vlan C) \vlor (B \vlan
D))}{
        \vlin{\wrD}{}{(B \vlan D) \vlor (\forall x.(A \vlan C))}{
          \vlhy{B \vlan D}}}}}
  \leadsto
  \vlderivation{
    \vlin{\fequ}{}{\forall x.((A \vlor B) \vlan (C \vlor D))}{
      \vlin{\wrD}{}{\forall x.((B \vlor A) \vlan (D \vlor C))}{
        \vlin{\wrD}{}{\forall x.((B \vlor A) \vlan D)}{
          \vlin{\faD}{}{\forall x.(B \vlan D)}{
            \vlhy{B \vlan D}}}}}}$}
  \end{equation*}
  where in the second case, $x$ is free in $B \vlan D$.

\item $\wrD/\fequ/\mfaD$:
  \begin{equation*}
  \vlderivation{
    \vlin{\mfaD}{}{\forall x.(A \vlor B)}{
      \vlin{\fequ}{}{(\forall x.A) \vlor (\forall x.B)}{
        \vlin{\wrD}{}{(\forall x.B) \vlor (\forall x.A)}{
          \vlhy{\forall x.B}}}}}
  \leadsto
  \vlderivation{
    \vlin{\fequ}{}{\forall x.(A \vlor B)}{
      \vlin{\wrD}{}{\forall x.(B \vlor A)}{
        \vlhy{\forall x.B}}}}
  \end{equation*}
  
  \begin{equation*}
  \vlderivation{
    \vlin{\mfaD}{}{\forall x.(A \vlor B)}{
      \vlin{\fequ}{}{(\forall x.A) \vlor (\forall x.B)}{
        \vlin{\wrD}{}{\forall x.((\forall x.A) \vlor B)}{
          \vlhy{\forall x.\forall x.A}}}}}
  \leadsto
  \vlderivation{
    \vlin{\wrD}{}{\forall x.(A \vlor B)}{
      \vlin{\cfaD}{}{\forall x.A}{
        \vlhy{\forall x.\forall x.A}}}}
  \end{equation*}
  
\item $\wrD/\fequ/\mexD$:

  \begin{equation*}
  \vlderivation{
    \vlin{\mexD}{}{\exists x.(A \vlor B)}{
      \vlin{\fequ}{}{(\exists x.A) \vlor (\exists x.B)}{
        \vlin{\wrD}{}{(\exists x.B) \vlor (\exists x.A)}{
          \vlhy{\exists x.B}}}}}
  \leadsto
  \vlderivation{
    \vlin{\fequ}{}{\exists x.(A \vlor B)}{
      \vlin{\wrD}{}{\exists x.(B \vlor A)}{
        \vlhy{\exists x.B}}}}
  \end{equation*}

\end{itemize}


\subsection{Proof of Lemma~\ref{lem:surj->cm}}
\label{app:surj->cm}

\begin{proof}[Proof of Lemma~\ref{lem:surj->cm}]
  By
  \cite[Proposition~7.5]{str:ral:tableaux19}, there is a derivation
  $\vlder{\set{\acD,\me,\PE\fequ}}{\DDeri}{\PE H}{\PE{(G\rsubstof\phi)}}$,
   We plan to show that $\DDeri$ can be lifted to
  $\set{\acD,\cfaD,\me,\mfaD,\mexD,\fequ}$. However, observe that not
  every formula occurring in $\DDeri$ is a propositional
  encoding. There are two reasons for this: (i) we might have
  $P\PE\fequ Q$ where $P$ is a propositional encoding but $Q$ is not,
  and (ii) the rule $\acD$ can duplicate an atom
  $x\in\VAR$. Let us write $\acDx$ for such instances.
  The problem with (i) is that we could have the following situation
  \begin{equation}
    \label{eq:m-illegal}
    \vlderivation{
      \vlin{\me}{}{\Scons{((x\vlan E)\vlor(x\vlan F))\vlan(C\vlor D)}}{
        \vlin{\PE\fequ}{}{\Scons{((x\vlan E)\vlan C)\vlor((x\vlan F)\vlan D)}}{
          \vlhy{\Scons{(x\vlan (E\vlan C))\vlor(x\vlan (F\vlan D))}}}}}
  \end{equation}
  where $x$ occurs in $C\vlor D$. Then premise and conclusion are both
  propositional encodings, but the whole derivation cannot be
  lifted. However, since we demand that the mapping is a fibration
  (and therefore a homomorphism) on the binding graphs, there must be
  another instance of $\me$ further below in the derivation:
  \vspace{-2ex}\begin{equation}
    \vlderivation{
      \vlin{\me}{}{S'\cons{(x\vlor x)\vlan(E\vlor F)}}{
        \vlhy{S'\cons{(x\vlan E)\vlor(x\vlan F)}}}}
  \end{equation}
  We can permute both instances via the following more general scheme
  (see~\cite{str:07:RTA,lamarche:gap} for a general discussion on
  permutations of the $\me$-rule):
  \begin{equation}
    \label{eq:m-permutation}
    \scalebox{.66}{$
    \vlderivation{
      \vlin{\me}{}{\Scons{(G\vlor H)\vlan(E\vlor F)\vlan(C\vlor D)}}{
        \vlin{\me}{}{\Scons{((G\vlan E)\vlor(H\vlan F))\vlan(C\vlor D)}}{
          \vlhy{\Scons{(G\vlan E\vlan C)\vlor(G\vlan F\vlan D)}}}}}
    \;
    \leftrightarrow
    \;
    \vlderivation{
      \vlin{\me}{}{\Scons{(G\vlor H)\vlan(E\vlor F)\vlan(C\vlor D)}}{
        \vlin{\me}{}{\Scons{(G\vlor H)\vlan((E\vlan C)\vlor(F\vlan D))}}{
          \vlhy{\Scons{(G\vlan E\vlan C)\vlor(G\vlan F\vlan D)}}}}}    
    $}
  \end{equation}
  We omitted some instances of $\PE\fequ$ and some parentheses.  We
  now call instances of $\me$ as in~\eqref{eq:m-illegal}
  \bfit{illegal}, and we can transform $\DDeri$ through
  $\me$-permutations~\eqref{eq:m-permutation} into a derivation that
  does not contain any illegal $\me$-instances.
  To
  address (ii),
  we also apply a permutation argument, permuting all
  instances of $\acDx$ up until they either reach the top of the
  derivation or an instance of $\me$ which separates the two atoms in
  the premise. More precisely, we consider the following inference rule
  \begin{equation}
    \label{eq:acDeq}
    \vlinf{\acDeq}{}{\Scons{x}}{S_0\cons{S_1\cons{x}\vlor S_2\cons{x}}}
  \end{equation}
  where $S_1\conhole\fequ\conhole\vlor E$ and
  $S_2\conhole\fequ\conhole\vlor F$ and $S\conhole\fequ
  S_0\cons{\conhole\vlor E\vlor F}$ for some formulas $E$ and $F$, where $E$ or $F$ or both might be empty. The
  rule $\acDeq$ permutes over $\fequ$, $\acD$, and other instances of $\acDeq$,
  and over instances of $\me$ if they occur inside $S_0$ or $S_1$ or
  $S_2$. The only situation in which $\acDeq$ cannot be permuted up is
  the following:
  \begin{equation}
    \label{eq:mac}
    \vlderivation{
      \vlin{\acDeq}{}{\Scons{R\cons{x}\vlan(C\vlor D)}}{
        \vlin{\me}{}{\Scons{(R_1\cons{x}\vlor R_2\cons{x})\vlan(C\vlor D)}}{
          \vlhy{\Scons{(R_1\cons{x}\vlan C)\vlor (R_2\cons{x}\vlan D)}}}}}
  \end{equation}
  We can therefore assume that all instances of $\acDx$, that contract
  an atom $x\in\VAR$ are either at the top of $\DDeri$ or below a
  $\me$-instance as in~\eqref{eq:mac}. We now lift $\DDeri$ to
  $\set{\acD,\cfaD,\me,\mfaD,\mexD,\fequ}$, proceed by induction on
  the height of $\DDeri$, beginning at the top, making a case
  analysis on the topmost rule that is not a~$\fequ$.
  \begin{itemize}
  \item $\acDx$: We know that the premise of~\eqref{eq:acDeq} is a
    propositional encoding. Hence, $S_1\conhole=\conhole\vlor \PE E$ and
    $S_2\conhole=\conhole\vlor \PE F$ and both $x$ are universals, and
    $\PE E\vlor \PE F$ contains all occurrences of $x$ bound by that
    universal. We have the following subcases:
    \begin{itemize}
    \item $E$ and $F$ are both non-empty: We have
      \begin{equation*}
        \vlinf{\acDeq}{}{\PE S\cons{x\vlor (\PE E\vlor \PE F)}}{\PE S\cons{(x\vlor \PE E)\vlor (x\vlor\PE F)}}
      \end{equation*}
      which can be lifted to
      \begin{equation*}
        \vlinf{\mfaD}{}{\Scons{\forall x.(E\vlor F)}}{S\cons{(\forall x.E)\vlor (\forall x.F)}}
      \end{equation*}
      where $\PE S\conhole, \PE E, \PE F$ are the propositional encodings of $\Sconhole,E,F$, respectively.
    \item $\PE E$ is empty and $\PE F$ is non-empty: We have
      \begin{equation*}
        \vlinf{\acDeq}{}{\PE S\cons{x\vlor \PE F)}}{\PE S\cons{x\vlor (x\vlor \PE F)}}
      \end{equation*}
      which can be lifted to
      \begin{equation*}
        \vlinf{\cfaD}{}{\Scons{\forall x.F}}{S\cons{\forall x.\forall x.F}}
      \end{equation*}
    \item $\PE E$ is non-empty and $\PE F$ is empty: This is similar to the previous case.
    \item $\PE E$ and $\PE F$ are both empty: This is impossible as the
      premise would not be a propositional encoding.
    \end{itemize}
  \item $\acD$ (contracting an ordinary atom): This can trivially be lifted.
  \item $\me$: There are several cases to consider.
    \begin{itemize}
    \item If none of the four principal formulas in the premise is $x$
      or $x\vlor F$ for some formula $F$ and $x\in\VAR$, then this
      instance of $\me$ can trivially be lifted, and we can
      proceed by induction hypothesis.
  \item If exactly one of the four principal formulas in the premise
    is $x$ for some $x\in\VAR$, then this $x$ is the encoding of an existential in the
    premise and of an universal in the conclusion. This is impossible,
    as $\phi$ has to preserve existentials.
  \item If two of the four principal formulas in the premise
    are $x$ for some $x\in\VAR$, then we are in the following special case of~\eqref{eq:mac}:
    \begin{equation*}
      \vlderivation{
        \vlin{\acDeq}{}{\Scons{x\vlan(C\vlor D)}}{
          \vlin{\me}{}{\Scons{(x\vlor x)\vlan(C\vlor D)}}{
            \vlhy{\Scons{(x\vlan C)\vlor (x\vlan D)}}}}}
    \end{equation*}
    which can be lifted immediately to
    \begin{equation*}
      \vlinf{\mexD}{}{\Scons{\exists x.(C\vlor D)}}{\Scons{(\exists x.C)\vlor(\exists x.D)}}
    \end{equation*}
  \item We have a situation~\eqref{eq:mac} where
    $R_1\cons{x}\fequ x\vlor E$ for some $E$ and $R_2\cons{x}\fequ
    x\vlor F$ for some $F$ with $R\cons{x}\fequ x\vlor E\vlor
    F$ (Otherwise, the application of $\acDeq$ would not be correct.)
    That means, we have: 
    \begin{equation*}
      \vlderivation{
        \vlin{\acDeq}{}{\Scons{(x\vlor E\vlor F)\vlan(C\vlor D)}}{
          \vlin{\me}{}{\Scons{((x\vlor E)\vlor (x\vlor F))\vlan(C\vlor D)}}{
            \vlhy{\Scons{((x\vlor E)\vlan C)\vlor ((x\vlor F)\vlan D)}}}}}
    \end{equation*}
    which can be lifted to
    \begin{equation*}
      \vlderivation{
        \vlin{\mfaD}{}{\Scons{(\forall x.(E\vlor F))\vlan(C\vlor D)}}{
          \vlin{\me}{}{\Scons{((\forall x. E)\vlor (\forall x. F))\vlan(C\vlor D)}}{
            \vlhy{\Scons{((\forall x. E)\vlan C)\vlor ((\forall x. F)\vlan D)}}}}}
    \end{equation*}
  \item In all other cases (e.g.\ exactly one of the principal formulas
    is of shape $x\vlor F$ (and none is $x$), we can trivially lift
    the $\me$-instance, as the quantifier structure is not affected.
    \end{itemize}
  \end{itemize}
  Thus $\DDeri$ can be lifted to
  $\vlder{\set{\acD,\cfaD,\me,\mfaD,\mexD,\fequ}}{\Deri}{H}{G\rsubstof\phi}$. 
\end{proof}
\end{document}